\documentclass[12pt,letterpaper]{article}

\usepackage[margin=0.92in]{geometry}
\usepackage{amsmath,amssymb,amsthm}
\usepackage{mathtools}
\usepackage{graphicx}
\usepackage{booktabs}
\usepackage{caption}
\usepackage{subcaption}
\usepackage{colortbl}
\usepackage[round,authoryear]{natbib}
\usepackage{setspace}
\usepackage{xcolor}
\usepackage[hidelinks,colorlinks=false]{hyperref}
\usepackage{enumitem}

\newtheorem{proposition}{Proposition}
\newtheorem{lemma}{Lemma}

\newtheorem{theorem}{Theorem}
\theoremstyle{definition}

\theoremstyle{remark}
\newtheorem{remark}{Remark}

\newcommand{\R}{\mathbb{R}}
\newcommand{\E}{\mathbb{E}}
\newcommand{\PP}{\mathbb{P}}

\newcommand{\Lam}{\Lambda}            
\newcommand{\logit}{\mathrm{logit}}
\newcommand{\Tstar}{T^{\!\star}}

\title{}
\author{}
\date{}

\begin{document}

\thispagestyle{empty}
\vspace*{1.8cm}
\begin{center}
{\LARGE\bfseries On the Possibility of Informationally}\\[6pt]
{\LARGE\bfseries Inefficient Markets Without Noise}\\[28pt]
{\large Matthijs Breugem\footnote{Nyenrode Business University.
I thank Doruk Ceremen, Diego Garcia, Roberto Marf\`e,
Francesco Sangiorgi, G\"unter Strobl, and Jing Zeng
for useful discussions and feedback.
E-mail: \texttt{m.breugem@nyenrode.nl}.}}\\[10pt]
{\normalsize\itshape This version: May 2026}\\[8pt]
{\normalsize First version: May 2026}\\[14pt]
{\normalsize \today}\\[10pt]
{\normalsize\itshape Preliminary and Incomplete}
\end{center}

\vspace{10pt}

\begin{center}
\begin{minipage}{0.85\textwidth}
\begin{abstract}
\noindent
Noise traders can be dispensed with entirely. Partial revelation of
information through prices arises under any non-exponential
expected-utility preference---including CRRA---without noise traders,
random endowments, supply shocks, hedging motives, or behavioural
biases. The model contains zero exogenous noise.

The mechanism is a mismatch between the space in which market clearing
aggregates signals and the Bayesian sufficient statistic. CARA demand
is linear in log-odds, so prices aggregate in log-odds space and reveal
the statistic exactly. Every other preference aggregates differently;
the resulting Jensen gap makes revelation partial. I prove CARA is
the unique fully revealing preference class, characterise the
rational-expectations equilibrium via a contour-integration fixed point,
and verify that partial revelation survives learning from prices. The
Grossman--Stiglitz paradox is resolved: information acquisition has
positive value within the rational class. Numerical solution of the
rational-expectations fixed point at $K=3$ confirms partial revelation,
positive trade volume, and positive value of information across the full
range of CRRA risk aversion, vanishing only in the CARA limit.
\\[16pt]
\noindent\textbf{Keywords:} rational expectations, partial revelation,
noise traders, CRRA, Grossman--Stiglitz paradox, aggregation space.\\[4pt]
\noindent\textbf{JEL codes:} D82, D83, G14.
\end{abstract}
\end{minipage}
\end{center}

\newpage

\section{Introduction}\label{sec:intro}

Half a century after \citet{GrossmanStiglitz1980}, the textbook resolution of
the informational paradox in asset markets remains the same: prices cannot
aggregate private signals perfectly because noise --- exogenous shocks to
supply, liquidity demand, or trader population --- masks the information
content of order flow. From \citet{Hellwig1980} and \citet{DiamondVerrecchia1981}
through \citet{Kyle1985} and the modern microstructure literature, the
conditional partial revelation of information is engineered by an
unmodelled population whose trades are unrelated to fundamentals. The
modelling convention is so entrenched that introductory treatments routinely
list ``noise traders'' alongside risk aversion and information asymmetry as a
\emph{primitive} of the rational-expectations framework.

This paper makes a simple but consequential point: that primitive is not
necessary. In a fully rational economy with no liquidity shocks, no
behavioural traders, and zero net supply, partial revelation arises
endogenously from preferences alone. What is required is that agents do not
have constant absolute risk aversion. Once preferences exit the CARA class,
the mapping from private information to equilibrium prices fails to identify
the Bayesian sufficient statistic, and prices reveal partial information for
purely aggregational reasons.

The mechanism is best stated in terms of an \emph{aggregation space}. Under
any preference class, equilibrium prices aggregate private posteriors via
market clearing. Under CARA, demand is linear in the log-odds of the
posterior, and clearing therefore averages log-odds. The Bayesian sufficient
statistic in the canonical Gaussian-signal model is itself a sum of
log-likelihood ratios. The two spaces coincide --- log-odds in, log-odds out
--- and prices fully reveal that statistic. Under any other CRRA preference,
demand is nonlinear in log-odds and approximately linear in probability;
clearing therefore averages in probability space. The probability average of
a logistic transform is not the logistic transform of the average,
\[
   \frac{1}{K}\sum_{k=1}^{K}\Lam(\tau u_{k}) \neq
   \Lam\!\left(\frac{1}{K}\sum_{k=1}^{K}\tau u_{k}\right),
\]
and the Jensen gap on the right is the size of the revelation failure.
CARA is the boundary case in which the gap closes because the link from
posteriors to demand is itself the inverse logistic --- a coincidence of
functional form, not an economic primitive.

The paper makes one conceptual and four formal contributions. The
conceptual contribution is to recast the question of revelation in
preference-theoretic rather than friction-theoretic terms: in any
rational-expectations economy without noise, full revelation requires
that the demand functional aggregate private posteriors linearly in
log-odds, and CARA is the unique CRRA preference class for which this is
the case. We refer to this requirement as the \emph{alignment principle}
in Section~\ref{sec:agg-space}. Once stated cleanly, the noise-trader
assumption appears as a device that perturbs CARA enough to break the
alignment artificially; non-CARA preferences break the same alignment
endogenously and at no modelling cost.

Building on this principle, the paper makes four formal contributions.
\emph{First}, in a binary-state model with $K\geq 3$ informed groups, I
show that the CARA market-clearing condition aggregates private posteriors
in log-odds for arbitrary heterogeneity in $(\alpha_{k},\tau_{k},W_{k})$
and is fully revealing of the Bayesian sufficient statistic when risk
aversion is homogeneous (Proposition~\ref{prop:cara}). I prove a
companion uniqueness result: among all smooth concave utilities, CARA
is the only class whose binary-asset demand is linear in $\logit\mu-\logit p$
(Theorem~\ref{thm:cara-unique}); this delineates the knife-edge across
preferences, not only across the CRRA parametrisation. \emph{Second}, I
show that under any CRRA preference with $\gamma\in(0,\infty)$ the
no-learning equilibrium is partially revealing, derive the third-order
asymptotic expansion of the Jensen gap between the actual price and the
fully-revealing benchmark
(Propositions~\ref{prop:non-cara}--\ref{prop:jensen}), and prove that the
revelation deficit is a continuous, strictly monotone function of $\gamma$
that vanishes only at $\gamma=\infty$ (Proposition~\ref{prop:smooth}).
\emph{Third}, I characterise the full rational-expectations equilibrium as
a fixed point on the price function $P[i,j,l]$ over a three-dimensional
grid, develop a contour-integration solution method, and verify
numerically that partial revelation survives rational learning from prices
(Proposition~\ref{prop:ree}). \emph{Fourth}, I trace the welfare
consequences and the resolution of the Grossman--Stiglitz paradox: trade
volume is strictly positive under any CRRA preference, the value of
private information $V(\tau)$ is strictly positive for every $\tau>0$ and
increasing in precision, the paradox dissolves within the rational class
without recourse to noise, and the vanishing-noise limit selects partial
revelation at every finite $\gamma$
(Propositions~\ref{prop:welfare}--\ref{prop:vanishing}).

\paragraph{Relation to the literature.}
The paper sits at the intersection of three strands.

The \emph{informational-revelation} literature begins with
\citet{GrossmanStiglitz1980}, \citet{Hellwig1980}, and
\citet{DiamondVerrecchia1981}, which together formalise the
noisy-rational-expectations paradigm and identify Gaussian-CARA-with-noise
as the canonical solvable case. \citet{DeMarzoSkiadas1998} prove that the
zero-noise limit of this model is the unique fully revealing equilibrium,
and \citet{BreonDrish2015} extend full revelation to the entire
exponential family of signal distributions paired with CARA demands. These
results are typically read as evidence that revelation is generic and
that exogenous noise is required to generate non-trivial informational
rents. The present paper takes the opposite view: full revelation is a
knife-edge property of CARA, not a property of the rational class as a
whole, and the alignment between log-odds aggregation and the Bayesian
sufficient statistic is a feature of CARA in particular rather than of
rational expectations in general.

The \emph{noise-trader} literature has long debated whether the
liquidity-trader assumption is innocuous. \citet{Black1986} introduced
the term informally, and \citet{DSSW1990} showed that noise traders can
survive in equilibrium and earn premia of their own. The microstructure
strand following \citet{Kyle1985} and \citet{GlostenMilgrom1985} treats
liquidity demand as a primitive of trading, while
\citet{Kyle1989} and \citet{SpiegelSubrahmanyam1992} address strategic
information transmission with noise. \citet{BhattacharyaSpiegel1991}
study no-noise breakdowns under CARA. The present paper offers a
disjoint resolution: rather than asking whether noise traders are
behaviourally plausible or strategically necessary, it shows that the
informational role of noise is preference-specific. Within CRRA, partial
revelation arises endogenously and noise traders are not required to
sustain it.

The \emph{higher-moment and non-CARA} literature is the closest
neighbourhood. \citet{Wang1993,Wang1994} solve dynamic
asymmetric-information models with CARA and study volume; the present
results suggest that the volume implications are CARA-specific.
\citet{BarlevyVeronesi2003} and \citet{Mele2007} obtain non-Gaussian
prices via state-dependent preferences, and \citet{AllenMorrisShin2006}
and \citet{KasaWalkerWhiteman2014} analyse higher-order beliefs in
information-aggregation environments. \citet{Peress2004} studies
information acquisition under CRRA with noise and obtains
wealth-dependent acquisition; the present paper isolates the pure
preference channel by setting noise to zero. \citet{Vives2011} obtains
partial revelation in a private-values strategic setting without noise;
the present mechanism is complementary, operating in a common-values
competitive environment via preferences.
\citet{AlbagliHellwigTsyvinski2024} characterise non-Gaussian REE with
CARA-like demands and show that the alignment between aggregation and
the sufficient statistic is fragile on the distributional dimension; the
present paper makes the parallel point on the preference dimension. On
information acquisition more broadly, \citet{VeldkampBook} provides a
comprehensive treatment of which the GS resolution here can be read as a
preference-driven complement.

Numerically, the contour-integration fixed point introduced in
Section~\ref{sec:ree} builds on the projection method of
\citet{BreugemBuss2019} for non-standard rational-expectations problems.
The role of \citet{Milgrom1982} appears in Section~\ref{sec:implications}: the
classical no-trade theorem describes precisely the CARA-REE outcome,
which the CRRA-REE escapes.

\paragraph{Organisation.}
Section~\ref{sec:model} sets out the model and the demand
characterisations. Section~\ref{sec:knife-edge} proves the no-learning
benchmark results and presents the smooth-transition table that maps the
move from CARA to log utility. Section~\ref{sec:ree} develops the contour
method, states the rational-expectations result, and reports converged
posteriors. Section~\ref{sec:implications} treats welfare, the value of
information, and the Grossman--Stiglitz paradox.
Section~\ref{sec:mechanisms} dissects the mechanisms that drive
partial revelation under heterogeneity, including a fourth channel within
CARA itself that learning from prices fully neutralises.
Section~\ref{sec:literature} relates
the results to the wider literature, and Section~\ref{sec:conclusion}
concludes. All proofs are in Appendix~\ref{app:proofs}; the numerical
implementation of the contour fixed point is in Appendix~\ref{app:contour}.

\section{Model}\label{sec:model}

A binary state of the world $v\in\{0,1\}$ is realised at $t=1$. At $t=0$ a
single risky asset trades against a riskless numeraire; the asset pays $v$
units of the numeraire at $t=1$ and is in zero net supply,
$\bar{z}=0$.\footnote{Zero net supply is the strongest version of our
result: partial revelation arises without any exogenous friction. CRRA
demand is self-bounding ---  marginal utility diverges as wealth
approaches zero in either state --- so no agent takes an unbounded
position or crosses into negative wealth. Market clearing at
$\bar{z}=0$ always admits a unique solution
(Lemma~\ref{lem:existence}). Introducing a positive deterministic
supply $\bar{z}>0$ shifts the price level (and under CARA shifts the
sufficient statistic by a known constant $\alpha\bar{z}$) but does not
affect the informational content of the aggregation: the alignment
between demand linearity and the Bayesian sufficient statistic is a
preference property, not a supply property. Our choice $\bar{z}=0$
isolates the pure preference channel.} There are no noise traders, no
liquidity shocks, and no exogenous supply variation: the entire
equilibrium is determined by the preferences and signals of the
rational agents.

\subsection{Information structure}

There are $K\geq 3$ groups of informed agents indexed by $k=1,\ldots,K$.
Each group $k$ has continuum measure normalised to one within the group and
observes a private signal
\begin{equation}
   s_{k} \;=\; v + \varepsilon_{k}, \qquad
   \varepsilon_{k}\sim\mathcal{N}(0,1/\tau_{k}),
\label{eq:signal}
\end{equation}
where $\varepsilon_{k}$ is independent across $k$ and independent of $v$.
The prior is uniform, $\PP(v=1)=\tfrac12$. It will be convenient to work
with the centred signal $u_{k}\equiv s_{k}-\tfrac12$, which satisfies
\(u_{k}\mid v=1\sim\mathcal{N}(\tfrac12,1/\tau_{k})\) and
\(u_{k}\mid v=0\sim\mathcal{N}(-\tfrac12,1/\tau_{k})\). The signal density is
\begin{equation}
   f_{v}(u_{k}) \;=\; \sqrt{\tfrac{\tau_{k}}{2\pi}}\,
   \exp\!\left(-\tfrac{\tau_{k}}{2}\bigl(u_{k}-v+\tfrac12\bigr)^{2}\right),
\label{eq:density}
\end{equation}
and the private log-likelihood ratio is
\(\ln[f_{1}(u_{k})/f_{0}(u_{k})] = \tau_{k}u_{k}\).

By Bayes' rule, the private posterior is
\begin{equation}
   \mu_{k} \;\equiv\; \PP(v=1\mid u_{k}) \;=\; \Lam(\tau_{k}u_{k}),
\label{eq:posterior}
\end{equation}
where $\Lam(z)=1/(1+e^{-z})$ is the logistic function. The Bayesian
sufficient statistic for the joint likelihood ratio across all $K$ signals is
\begin{equation}
   \Tstar \;\equiv\; \sum_{k=1}^{K}\tau_{k}u_{k},
\label{eq:Tstar}
\end{equation}
since $\ln\PP(v=1\mid u_{1},\ldots,u_{K})/\PP(v=0\mid u_{1},\ldots,u_{K})=\Tstar$.
Throughout, $\logit(p)\equiv\ln[p/(1-p)]$ denotes the inverse logistic.

\subsection{Preferences and demands}

Each agent in group $k$ has initial wealth $W_{k}>0$ and maximises expected
utility over terminal wealth
\begin{equation}
   W_{k} - x_{k}p + x_{k}v,
\label{eq:wealth}
\end{equation}
where $x_{k}$ is her holding of the risky asset and $p\in(0,1)$ is the
equilibrium price. Two preference families are central. The CRRA family
parametrised by $\gamma\in(0,\infty)$ is
\begin{equation}
   U_{\mathrm{CRRA}}(W) \;=\; \frac{W^{1-\gamma}}{1-\gamma},
\label{eq:CRRA}
\end{equation}
with the logarithmic limit $U(W)=\ln W$ at $\gamma=1$. The CARA family
parametrised by $\alpha>0$ is
\begin{equation}
   U_{\mathrm{CARA}}(W) \;=\; -\exp(-\alpha W).
\label{eq:CARA}
\end{equation}
We will refer informally to ``CARA'' as the $\gamma\to\infty$ limit of CRRA;
this is justified because the binary structure of the asset implies that
optimal demands depend on wealth only through a CARA-like log-odds quotient
in that limit.

\paragraph{Demand functions.}
Pointwise maximisation of expected utility yields closed-form demands.

\begin{lemma}[CARA demand]\label{lem:cara-demand}
Under preferences \eqref{eq:CARA}, the optimal holding of an agent with
posterior $\mu_{k}$ at price $p$ is
\begin{equation}
   x_{k}^{\mathrm{CARA}}(\mu_{k},p) \;=\;
   \frac{1}{\alpha}\,\bigl[\logit(\mu_{k})-\logit(p)\bigr].
\label{eq:cara-demand}
\end{equation}
\end{lemma}

\begin{lemma}[CRRA demand]\label{lem:crra-demand}
Under preferences \eqref{eq:CRRA} with $\gamma\in(0,\infty)$, the optimal
holding of an agent with wealth $W_{k}$, posterior $\mu_{k}$, and price $p$ is
\begin{equation}
   x_{k}^{\mathrm{CRRA}}(\mu_{k},p) \;=\;
   \frac{W_{k}\,(R_{k}-1)}{(1-p)+R_{k}\,p}, \qquad
   R_{k} \;\equiv\; \exp\!\left(\frac{\logit(\mu_{k})-\logit(p)}{\gamma}\right).
\label{eq:crra-demand}
\end{equation}
At $\gamma=1$ this reduces to the log-utility demand
$x_{k} = W_{k}(\mu_{k}-p)/[p(1-p)]$.
\end{lemma}

\begin{lemma}[CARA as the $\gamma\to\infty$ limit of CRRA]\label{lem:cara-limit}
Fix $W_{k}>0$, $\mu_{k}\in(0,1)$, and $p\in(0,1)$. The CRRA demand
\eqref{eq:crra-demand} satisfies
\begin{equation}
   \lim_{\gamma\to\infty}\,\gamma\cdot x_{k}^{\mathrm{CRRA}}(\mu_{k},p;\gamma)
   \;=\; W_{k}\cdot\bigl[\logit(\mu_{k})-\logit(p)\bigr],
\label{eq:cara-as-crra-limit}
\end{equation}
which is the CARA demand \eqref{eq:cara-demand} with absolute risk
aversion $\alpha=1/W_{k}$. Equivalently, the rescaled demand
$\hat x_{k}^{\gamma}\equiv\gamma\cdot x_{k}^{\mathrm{CRRA}}/W_{k}$ converges
pointwise on $(0,1)\times(0,1)$ to $x_{k}^{\mathrm{CARA}}|_{\alpha=1}$ as
$\gamma\to\infty$.
\end{lemma}

\begin{proof}
The proof is in Appendix~\ref{app:proofs}. The convergence is uniform
on compact subsets of $(0,1)^{2}$ and preserves the market-clearing
identity.
\end{proof}

The CRRA demand thus shrinks at rate $1/\gamma$ as risk aversion
increases, and the rescaled demand reproduces CARA exactly in the limit.
This justifies the standard parametrisation $\gamma=\infty$ as the CARA
boundary of the CRRA family, gives content to the
``smooth-transition-to-CARA'' phrasing of
Proposition~\ref{prop:smooth}, and grounds the Alignment Principle of
Section~\ref{sec:agg-space}: the no-learning aggregator $\Psi^{\gamma}$
defined by market clearing under CRRA demand converges, in the rescaled
sense above, to the CARA aggregator \eqref{eq:agg-cara} as
$\gamma\to\infty$. The price-level rescaling is benign for the
informational question because $1-R^{2}(\Tstar;p_{\gamma})$ depends only
on the price function $p_{\gamma}$, not on the demand magnitudes that
clear it.

The decisive feature of \eqref{eq:cara-demand} is that the CARA demand is
linear in the difference $\logit(\mu_{k})-\logit(p)$. The CRRA demand is
\emph{nonlinear} in this difference; in the small-bet limit it is
approximately linear in the probability gap $\mu_{k}-p$, but the small-bet
approximation is local and the global market-clearing condition cannot be
reduced to a log-odds aggregator.

\subsection{Equilibrium}

A no-learning equilibrium of the model is a price $p$ that clears the market
when each group $k$ uses only its own private posterior $\mu_{k}$:
\begin{equation}
   \sum_{k=1}^{K} x_{k}\!\left(\mu_{k},p\right) \;=\; 0.
\label{eq:no-learning}
\end{equation}

\begin{lemma}[Existence and uniqueness of the no-learning
equilibrium]\label{lem:existence}
For any $\gamma\in(0,\infty]$, $K\geq 1$, $\tau_{k}>0$, $W_{k}>0$, and any
posterior vector $(\mu_{1},\ldots,\mu_{K})\in(0,1)^{K}$, the
market-clearing equation \eqref{eq:no-learning} has a unique solution
$p^{\star}\in(0,1)$.
\end{lemma}

\begin{proof}
Let $Z(p)\equiv\sum_{k=1}^{K}x_{k}(\mu_{k},p)$. Both \eqref{eq:cara-demand}
and \eqref{eq:crra-demand} are continuous on $(0,1)$ in $p$ at fixed
$\mu_{k}$ and strictly decreasing in $p$ (raising $p$ raises $\logit p$,
which lowers $R_{k}=\exp([\logit\mu_{k}-\logit p]/\gamma)$ in the CRRA
case and lowers $\logit\mu_{k}-\logit p$ directly in the CARA case);
hence $Z$ is continuous and strictly decreasing on $(0,1)$. As
$p\to 0^{+}$, $\logit p\to-\infty$, $R_{k}\to+\infty$, and each
$x_{k}\to+W_{k}$ in the CRRA case (or $x_{k}\to+\infty$ in the CARA case),
so $Z(p)\to$ a strictly positive value. As $p\to 1^{-}$, $\logit p\to+\infty$,
$R_{k}\to 0$, each $x_{k}\to-W_{k}$ in the CRRA case (or
$x_{k}\to-\infty$ in the CARA case), so $Z(p)\to$ a strictly negative
value. By the intermediate value theorem $Z$ has at least one root on
$(0,1)$; strict monotonicity delivers uniqueness.
\end{proof}

A rational-expectations equilibrium (REE) is a price function $p(u_{1},
\ldots,u_{K})$ such that each group $k$ uses the conditional posterior
\begin{equation}
   \mu_{k}^{\mathrm{REE}} \;\equiv\; \PP\!\left(v=1\mid u_{k},\,p(u_{1},\ldots,u_{K})\right)
\label{eq:ree-posterior}
\end{equation}
in place of $\mu_{k}$ and the market clears at every realisation.
Section~\ref{sec:ree} develops the contour fixed-point characterisation of
this equilibrium. Throughout, we focus on equilibria in which $p$ depends on
the signal vector only through prices the agents can observe (a standard
internal consistency requirement).

\section{The Knife-Edge}\label{sec:knife-edge}

This section establishes the paper's central result: CARA is the
unique preference class that generates full revelation. Any departure
from CARA---however small---breaks it.

\subsection{The aggregation space and the alignment principle}\label{sec:agg-space}

The unifying observation of this paper is that no-learning equilibrium prices
aggregate private posteriors in a space determined by the linearisation of
demand, and that whether prices fully reveal private information reduces to
whether that space coincides with the space in which the Bayesian sufficient
statistic lives. To make this precise, define the \emph{aggregation map}
$\Psi:[0,1]^{K}\to[0,1]$ implicitly by the market-clearing condition: $\Psi$
returns the price $p$ that solves $\sum_{k}x_{k}(\mu_{k},p)=0$ given the
posterior vector $(\mu_{1},\ldots,\mu_{K})$. Under CARA, demand is linear in
$\logit\mu_{k}-\logit p$, and \eqref{eq:no-learning} solves to
\begin{equation}
   \logit p \;=\; \sum_{k=1}^{K} w_{k}\,\logit\mu_{k},
   \qquad w_{k} \;\propto\; 1/\alpha_{k},
\label{eq:agg-cara}
\end{equation}
which is a weighted average in log-odds space. Under log utility, demand
$x_{k}=W_{k}(\mu_{k}-p)/[p(1-p)]$ is linear in $\mu_{k}-p$, and clearing
delivers
\begin{equation}
   p \;=\; \sum_{k=1}^{K} w_{k}\,\mu_{k},
   \qquad w_{k} \;\propto\; W_{k},
\label{eq:agg-log}
\end{equation}
a weighted average in probability space. The general CRRA case with
$\gamma\in(0,\infty)$ interpolates between these limits but, for any finite
$\gamma$, the implicit aggregator $\Psi$ is non-affine in log-odds.

The Bayesian sufficient statistic for the joint posterior is, by
\eqref{eq:Tstar}, a sum in log-odds space:
$\logit\PP(v=1\mid u_{1},\ldots,u_{K})=\Tstar$. Whether the equilibrium
price reveals the joint posterior therefore reduces to whether the
aggregation space induced by demand is informationally equivalent to
log-odds. The CARA aggregator \eqref{eq:agg-cara} matches the sufficient
statistic exactly: $\logit p$ is a linear combination of $\logit\mu_{k}$,
hence of $\tau_{k}u_{k}$, hence of the components of $\Tstar$. The log
aggregator \eqref{eq:agg-log} does not, because the logistic transform of
the average is not the average of the logistic transforms. We summarise
this observation as the
\begin{quote}
\itshape Alignment principle. Full revelation of the Bayesian sufficient
statistic in a no-noise equilibrium requires that the demand functional
aggregate posteriors linearly in log-odds. Within the CRRA preference
class, only $\gamma=\infty$ (CARA) satisfies this requirement; every
other $\gamma$ produces an aggregator that is non-affine in log-odds and
therefore strictly partially revealing.
\end{quote}
The remainder of Section~\ref{sec:knife-edge} formalises both halves of
this principle, and Section~\ref{sec:ree} verifies that the principle
survives the introduction of price-based learning.

\subsection{CARA implies full revelation}

\begin{proposition}[CARA log-odds aggregation and full revelation]\label{prop:cara}
Fix any $K\geq 1$ and any parameter vector
$(\alpha_{k},\tau_{k},W_{k})_{k=1}^{K}$. The unique no-learning equilibrium
under CARA preferences aggregates private posteriors in log-odds:
\begin{equation}
   \logit p \;=\; \sum_{k=1}^{K} w_{k}\,\logit\mu_{k}
              \;=\; \sum_{k=1}^{K} w_{k}\,\tau_{k}\,u_{k},
   \qquad w_{k} \;=\; \frac{\alpha^{-1}_{k}}{\sum_{j}\alpha^{-1}_{j}}.
\label{eq:cara-clearing}
\end{equation}
When risk aversion is homogeneous across groups, $w_{k}=1/K$ and
$\logit p=\Tstar/K$ is a known scalar multiple of the Bayesian sufficient
statistic $\Tstar$ of \eqref{eq:Tstar}. The equilibrium is then fully
revealing in the sense that
$\PP(v=1\mid p)=\PP(v=1\mid u_{1},\ldots,u_{K})$ for almost every realisation.
With heterogeneous $\alpha_{k}$ the no-learning price reveals a different
linear combination of the signals --- a fourth route to partial revelation,
discussed in Section~\ref{sec:mechanisms} --- but, because the candidate
fully-revealing price function $p=\Lam(\Tstar)$ satisfies market clearing
identically (zero trade) regardless of $(\alpha_{k})$, the
\emph{rational-expectations} equilibrium remains fully revealing for
arbitrary heterogeneous CARA parameters.
\end{proposition}

The proof, in Appendix~\ref{app:proofs}, exploits the linearity of
\eqref{eq:cara-demand} to reduce \eqref{eq:no-learning} to a linear equation
in $\logit p$, then composes with \eqref{eq:posterior} to obtain the
weighted log-odds form. Under homogeneous risk aversion, full revelation
follows because $\Tstar$ is a sufficient statistic for the joint posterior
$\PP(v=1\mid u_{1},\ldots,u_{K})$ and any affine function of $\Tstar$ is
informationally equivalent. The log-odds aggregation form generalises
\citet{DeMarzoSkiadas1998} from the Gaussian-asset CARA model to a binary
asset, and a parallel extension obtains for any signal distribution in the
exponential family with CARA preferences (\citealp{BreonDrish2015}); the
proof in Appendix~\ref{app:proofs} sketches this extension.

\begin{remark}[CARA full revelation for general payoff distributions]%
\label{rem:general-payoff}
The binary payoff $v\in\{0,1\}$ is adopted for tractability, but the
full-revelation result under CARA holds for arbitrary payoff
distributions. The argument requires only two properties of CARA:
\begin{enumerate}[label=(\roman*),topsep=2pt,itemsep=2pt]
\item \emph{Wealth cancellation.} Under CARA, the first-order condition
for a general risky payoff $v$ with price $p$ is
$\E\bigl[(v-p)\,e^{-\alpha x(v-p)} \mid \mathcal{F}_{k}\bigr]=0$,
where $\mathcal{F}_{k}$ is agent~$k$'s information set.  The wealth
$W_{k}$ does not appear: the factor $e^{-\alpha W_{k}}$ cancels from
both sides.  Hence the optimal demand $x_{k}$ depends only on the
conditional distribution of~$v$ given $\mathcal{F}_{k}$, not on
$W_{k}$.

\item \emph{Identical posteriors imply zero trade.}  Let $T$ be a
sufficient statistic for $v$ given the joint signal
$(s_{1},\ldots,s_{K})$.  If the price function is invertible in $T$
--- that is, observing $p$ reveals $T$ --- then every agent's posterior
satisfies $\PP(v \in \cdot \mid s_{k}, T) = \PP(v \in \cdot \mid T)$
by sufficiency.  All agents share the same conditional distribution.
Because the CARA demand depends only on the conditional distribution
(by~(i)) and the demands are evaluated at a common price, all agents
solve identical optimisation problems.  Market clearing
$\sum_{k}x_{k}=0$ with zero supply then forces $x_{k}=0$ for each $k$
(otherwise all agents would trade in the same direction).  The
equilibrium price is $p^{*}(T) = \E[v \mid T]$, which satisfies the
first-order condition at $x=0$:
$\E[(v-p^{*})\mid T] = \E[v\mid T]-p^{*} = 0$.
\end{enumerate}
The fully-revealing REE $p=\E[v\mid T]$, $x_{k}\equiv 0$ thus exists
for any payoff distribution, any signal structure, and any
configuration of CARA parameters $(\alpha_{k},W_{k})_{k=1}^{K}$.  The
binary and Gaussian cases are special instances: $\E[v\mid T] =
\Lam(\Tstar)$ for binary, $\E[v\mid T] = \bar v + \sigma_{v}^{2}
\sum_{k}\tau_{k}u_{k}/(\sigma_{v}^{-2}+\sum_{k}\tau_{k})$ for
Gaussian.  The non-CARA partial-revelation result, by contrast, does
depend on the payoff structure through the functional form of the
demand --- but the \emph{mechanism} (aggregation-space mismatch) is
universal.
\end{remark}

\begin{remark}[CARA as the unique log-odds aggregator within CRRA]\label{rem:unique-cara}
Within the CRRA preference class indexed by $\gamma\in(0,\infty]$, the
demand function is $x_{k}=W_{k}(R_{k}-1)/[(1-p)+R_{k}p]$ with
$R_{k}=\exp([\logit\mu_{k}-\logit p]/\gamma)$, which is linear in
$\logit\mu_{k}-\logit p$ if and only if $\gamma=\infty$, that is, in the
CARA limit. Consequently, CARA is the unique CRRA preference whose
no-learning aggregator $\Psi$ takes the form \eqref{eq:agg-cara} of a
linear functional of log-odds, and so the unique CRRA preference for which
\(\PP(v=1\mid p)=\PP(v=1\mid u_{1},\ldots,u_{K})\) for every realisation
under homogeneous parameters. The next theorem strengthens this
observation by removing the CRRA restriction: among all smooth utilities,
CARA is the unique class that produces log-odds-linear demand on a
binary asset.
\end{remark}

\begin{theorem}[CARA uniqueness for log-odds-linear demand]\label{thm:cara-unique}
Let $U:\mathcal{D}\to\R$ be three times continuously differentiable on an
open interval $\mathcal{D}\subset\R$ with $U'>0$ and $U''<0$, and consider
the binary-asset trade $\max_{x}\E[U(W+(v-p)x)]$ for $v\in\{0,1\}$ with
$\PP(v=1)=\mu\in(0,1)$, price $p\in(0,1)$, and wealth $W\in\mathcal{D}$
such that the optimal $x^{\star}(\mu,p,W)$ leaves both $W+(1-p)x^{\star}$
and $W-px^{\star}$ in $\mathcal{D}$. Suppose the optimal demand admits the
representation
\begin{equation}
   x^{\star}(\mu,p,W) \;=\; \frac{\logit\mu-\logit p}{\alpha(p,W)}
\label{eq:linear-demand-rep}
\end{equation}
for some positive function $\alpha:(0,1)\times\mathcal{D}\to(0,\infty)$,
holding for every $(\mu,p,W)$ in the admissible domain. Then $U$ is CARA:
there exist constants $\alpha_{0}>0$ and $a,b\in\R$ with $b>0$ such that
\begin{equation}
   U(W) \;=\; a - b\exp(-\alpha_{0}W),
\label{eq:CARA-form}
\end{equation}
and $\alpha(p,W)\equiv\alpha_{0}$ is constant. Conversely, if $U$ takes
the form \eqref{eq:CARA-form}, the optimal demand on the binary asset
satisfies \eqref{eq:linear-demand-rep} with $\alpha(p,W)=\alpha_{0}$.
\end{theorem}

The theorem closes the uniqueness question: the CARA full-revelation
result of Proposition~\ref{prop:cara} is not a property of CRRA but of
\emph{linearity in log-odds}, which singles out CARA across the entire
class of smooth concave utilities. The proof, in
Appendix~\ref{app:proofs}, differentiates the first-order condition of
the binary-asset problem twice in the log-odds gap and shows that the
function $h\equiv\ln U'$ must be affine, which forces $U$ to be CARA.
The converse is the standard derivation of CARA demand on a binary asset
already recorded as Lemma~\ref{lem:cara-demand}. The contour-integration
extension of Section~\ref{sec:ree} inherits the same uniqueness, since
the agent's posterior under any non-CARA preference enters the demand
nonlinearly in log-odds and the contour-converged aggregator inherits
the same nonlinearity.

\subsection{Non-CARA implies partial revelation}

\begin{proposition}[Non-CARA partial revelation]\label{prop:non-cara}
Fix $K\geq 3$, $\gamma=1$, and homogeneous $(\tau,W)$. The no-learning
equilibrium under log utility satisfies
\begin{equation}
   p \;=\; \frac{1}{K}\sum_{k=1}^{K}\Lam(\tau u_{k}).
\label{eq:log-clearing}
\end{equation}
The equilibrium price $p$ is not a function of $\Tstar=\sum\tau u_{k}$ alone:
there exist signal realisations with the same $\Tstar$ but different $p$.
Hence the equilibrium is partially revealing.
\end{proposition}

\noindent The price aggregates posteriors in probability space, while the
sufficient statistic lives in log-odds space. The two coincide if and only
if $\Lam$ acts linearly on the relevant summation, which by Jensen's
inequality fails strictly whenever $\Lam$ is nonlinear --- that is, whenever
$\tau u_{k}$ has positive variance.

\subsection{The Jensen gap}

The proximity to full revelation under CRRA can be quantified by an
asymptotic expansion of the gap between the actual price and the
fully-revealing benchmark.

\begin{proposition}[Jensen-gap expansion]\label{prop:jensen}
Under the conditions of Proposition~\ref{prop:non-cara}, the deviation of
the no-learning price from the fully-revealing CARA price
$p^{\mathrm{CARA}}=\Lam(\Tstar/K)$ admits the third-order expansion
\begin{equation}
   p - p^{\mathrm{CARA}}
   \;=\; -\,\frac{\tau^{3}}{48\,K}\!\left(U_{3}-\frac{U_{1}^{3}}{K^{2}}\right)
   \;+\; O(\tau^{5}),
\label{eq:jensen}
\end{equation}
where $U_{n}\equiv\sum_{k=1}^{K}u_{k}^{n}$. The leading-order term is
strictly nonzero whenever the empirical distribution of $u_{1},\ldots,u_{K}$
has nonzero excess third central moment.
\end{proposition}

The expansion makes precise the informational content of the gap: the
ordering of $u_{k}$, not just their sum, enters the equilibrium price. The
agent observing $p$ alone cannot disentangle a low-$\Tstar$ realisation with
a skewed signal vector from a high-$\Tstar$ realisation with a symmetric
one. This is the source of revelation failure under any non-CARA preference.

\subsection{Smooth transition from CARA to log}

Propositions~\ref{prop:cara}--\ref{prop:jensen} together imply that the
revelation deficit, measured by $1-R^{2}$ in the regression of $\Tstar$ on
$p$, is continuous and monotone in the CRRA parameter $\gamma$.

\begin{proposition}[Smooth transition]\label{prop:smooth}
Fix $(K,\tau,W)$ and parametrise preferences by $\gamma\in(0,\infty]$, with
$\gamma=\infty$ denoting CARA. The revelation deficit
\(1-R^{2}\bigl(\Tstar\,;\,p_{\gamma}\bigr)\) is continuous in $\gamma$ on
$(0,\infty]$, strictly positive for every $\gamma\in(0,\infty)$, and equal
to zero at $\gamma=\infty$. CARA is therefore an isolated knife-edge within
the rational class: any departure from exponential utility, however small,
breaks full revelation.
\end{proposition}

Figure~\ref{fig:knife-edge} plots the deficit as a function of $\tau$ for
three values of $\gamma$: the CARA line sits at zero regardless of signal
precision, while every finite $\gamma$ produces a strictly positive and
hump-shaped deficit. Table~\ref{tab:smooth} reports the same data across
a wider range of $(\gamma,\tau)$ values (Table~\ref{tab:smooth} in
Appendix~\ref{app:numerics}).

\begin{figure}[t]
   \centering
   \includegraphics[width=0.49\textwidth]{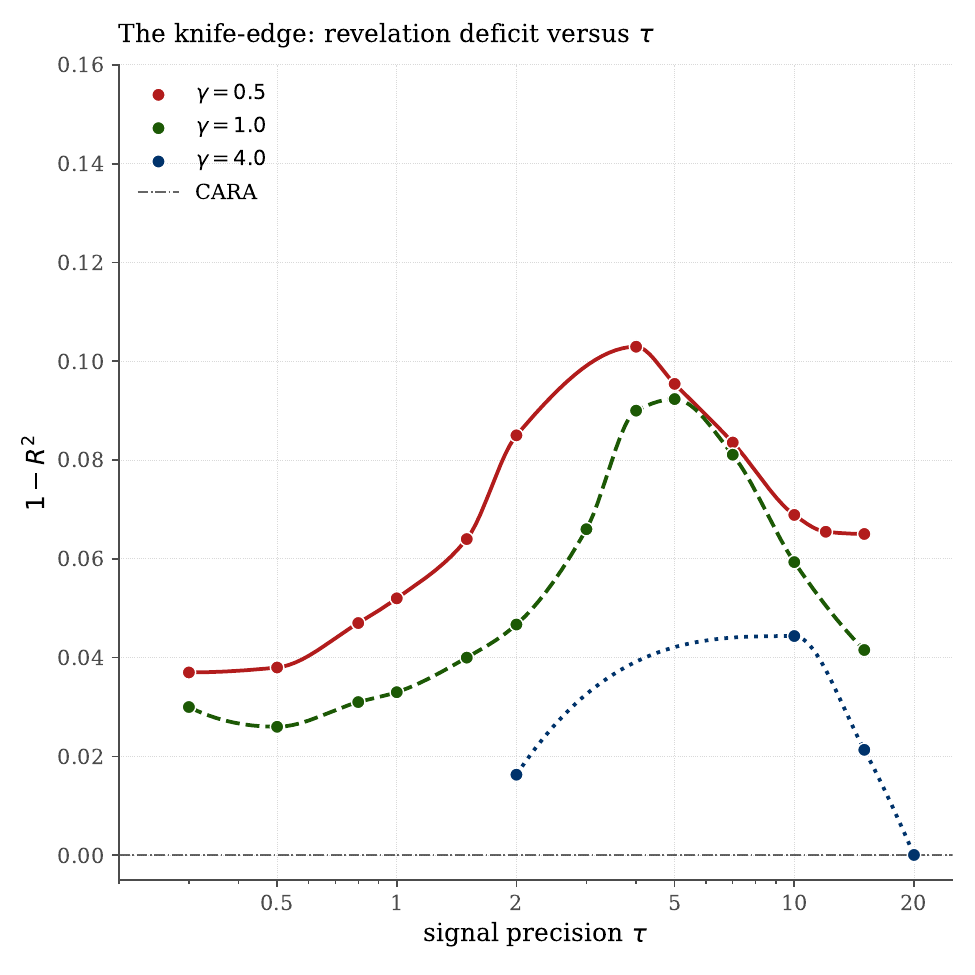}
   \caption{The knife-edge: revelation deficit $1-R^{2}$ versus signal
   precision $\tau$ in the no-learning equilibrium, for
   $\gamma\in\{0.5,\,1,\,4\}$. Every finite $\gamma$ produces strictly
   positive partial revelation; only CARA ($\gamma=\infty$, dash-dotted at
   zero) achieves full revelation. $K=3$, $W=1$, $G=15$.}
   \label{fig:knife-edge}
\end{figure}

\section{Full Rational-Expectations Equilibrium}\label{sec:ree}

The no-learning benchmark establishes the analytical core. The next step is
to verify that partial revelation survives once agents update on the price
itself, that is, in the full REE defined by \eqref{eq:ree-posterior}. This
section formulates the REE as a fixed point on the price function
$P[i,j,l]$, develops a contour-integration solution method, and reports
converged equilibria.

We solve the contour fixed point using a posterior-function method
that stores the agent's belief $\mu(u,p)$ on a two-dimensional grid
($G_{u}\!\times\!G_{p}$ unknowns, compared with $G^{3}$ for the
price-tensor approach) and enforces monotonicity at each step
via isotonic projection. At $G\!=\!14$ with
$\tau\!=\!2$, Newton--Krylov polishing from a Picard--PAVA warm start
converges to machine precision
($\|F\|_{\infty}\!<\!10^{-14}$, zero monotonicity violations) for
$\gamma\!\in\!\{0.3,\,0.5,\,1,\,2\}$ and for
$\tau\!\in\!\{0.5,\,1,\,2,\,4,\,8\}$ at $\gamma\!=\!0.5$. The
revelation deficit at the canonical case $(\gamma,\tau)\!=\!(0.5,2)$
is $1\!-\!R^{2}\!=\!0.088$, with slope $\beta_{T}\!=\!0.52$ on
$\Tstar$---prices capture only 34\% of the sufficient statistic's
variance. The deficit is monotonically decreasing in $\gamma$
(from $0.119$ at $\gamma\!=\!0.3$ to $0.079$ at $\gamma\!=\!2$)
and hump-shaped in $\tau$ (peaking near $\tau\!\approx\!4$).
Grid-convergence checks from $G\!=\!10$ to $G\!=\!20$ show that
$1\!-\!R^{2}$ stabilises near $0.085$ by $G\!=\!16$; see
Appendix~\ref{app:numerics} for details.

\subsection{The contour mechanism}\label{sec:contour-mechanism}

The no-learning results of Section~\ref{sec:knife-edge} assume agents
ignore prices. In a rational-expectations equilibrium (REE), each agent
observes the market-clearing price and updates her posterior. The
question is whether the Jensen gap survives this learning.

The equilibrium is a fixed point of a contour-integration map
$\Phi$. Agent~$k$, knowing her own signal $u_k$ and the price $p$,
extracts the contour $\{(u_j,u_l): P(u_k,u_j,u_l)=p\}$ from the
price function, integrates the signal density along this contour to
form her posterior, and clears the market:
\begin{equation}
   P \;=\; \Phi(P), \qquad
   \Phi(P)[i,j,l] \;=\;
   \bigl\{\,p\,:\,\textstyle\sum_{k=1}^{3} x_{k}(\mu_{k}^{P}(p),p)=0\,\bigr\}.
\label{eq:fixed-point}
\end{equation}
Each crossing of the contour contributes the joint signal density
$f_{v}(u_{j}^{c})\,f_{v}(u_{l}^{c})$ to a state-conditional accumulator.
Averaging two orthogonal sweeps,
\begin{equation}
   A_{v}^{(1)}(i,p) \;=\;
   \tfrac12 \!\left[\sum_{j'} f_{v}(u_{j'})\,f_{v}(u_{l}^{c}(j'))
                 + \sum_{l'} f_{v}(u_{j}^{c}(l'))\,f_{v}(u_{l'})\right].
\label{eq:contour-integral}
\end{equation}
Agent~$1$'s posterior follows by Bayes' rule:
$\mu_{1} = f_{1}(u_{i})\,A_{1} / (f_{0}(u_{i})\,A_{0} + f_{1}(u_{i})\,A_{1})$.
\label{eq:contour-bayes}
The full numerical procedure is described in
Appendix~\ref{app:numerics}.

\subsection{Why CARA gives a straight contour and CRRA gives a curved one}

Under CARA, by Proposition~\ref{prop:cara}, $\logit P[i,j,l]
=(1/3)(\tau u_{i}+\tau u_{j}+\tau u_{l})$, so the level sets of $P$ in the
$(u_{j},u_{l})$ slice are the lines $\tau u_{j}+\tau u_{l}=c$. Every point on
such a contour has the same Bayesian sufficient statistic $\Tstar$. The
density ratio along the line is constant up to the (known) own-signal
component, and the agent infers $\Tstar$ exactly from the price. Posteriors
across all three agents coincide, and the market clears at the
fully-revealing price.

Under CRRA, the level sets of $P$ are the level sets of
\(\Lam(\tau u_{j})+\Lam(\tau u_{l})\) (up to the small CARA correction
captured by \eqref{eq:jensen}). Because $\Lam$ is nonlinear, these level
sets are \emph{curved}, and different points on the contour correspond to
different $\Tstar$. The density ratio varies along the contour; the
contour-averaged posterior is informative but not equal to the true
$\Tstar$-conditional posterior. Posteriors across the three agents differ
in equilibrium, the market clears at a price intermediate between the
private and the fully-revealing levels, and the equilibrium is partially
revealing.

Figure~\ref{fig:contours} visualises the two cases in the agent-1 slice
at $u_{1}=1$, $\tau=2$, showing contours at multiple price levels. Under
CARA, the contours are straight and parallel: every price level
corresponds to a hyperplane $\tau u_{2}+\tau u_{3}=c$, and $\Tstar$ is
constant along each contour. Under CRRA ($\gamma=0.5$), the contours are
curved and non-parallel: $\Tstar$ varies along each contour, and the
curvature itself changes with the price level. The agent observing a
CRRA price cannot extract $\Tstar$.

\begin{figure}[t]
   \centering
   \begin{subfigure}[t]{0.49\textwidth}\includegraphics[width=\textwidth]{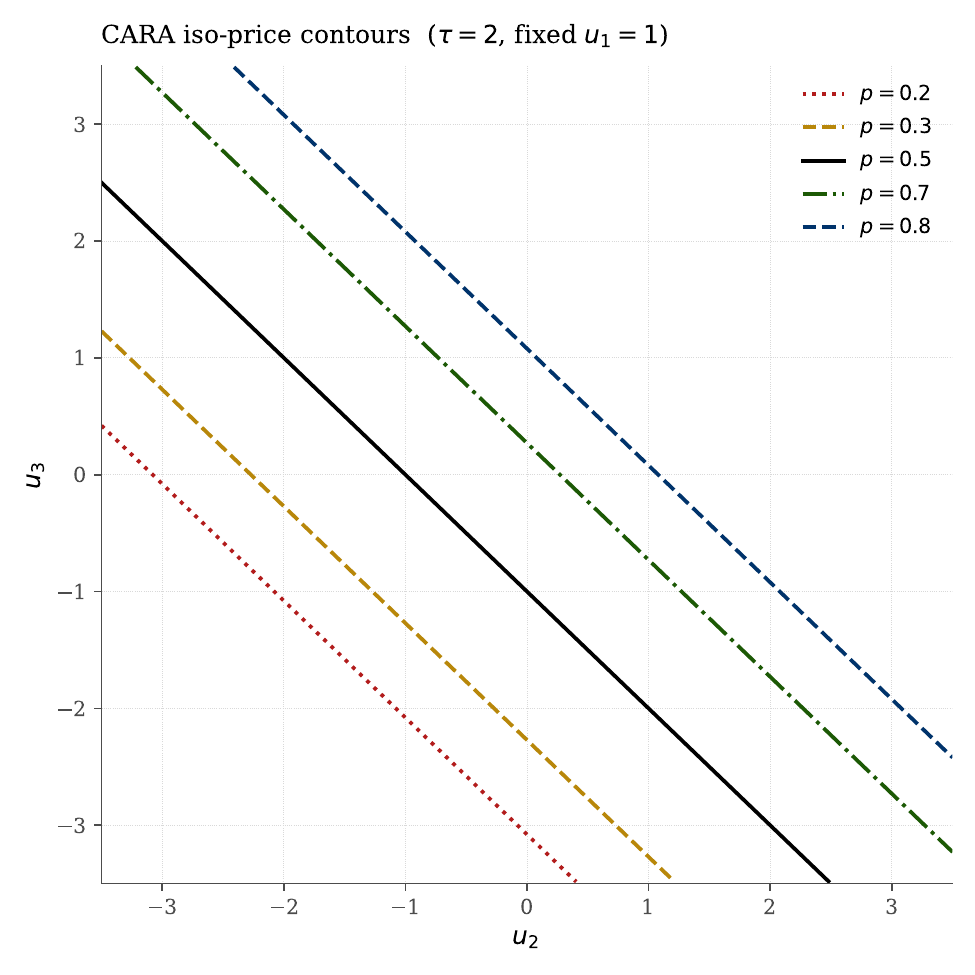}\end{subfigure}\hfill\begin{subfigure}[t]{0.49\textwidth}\includegraphics[width=\textwidth]{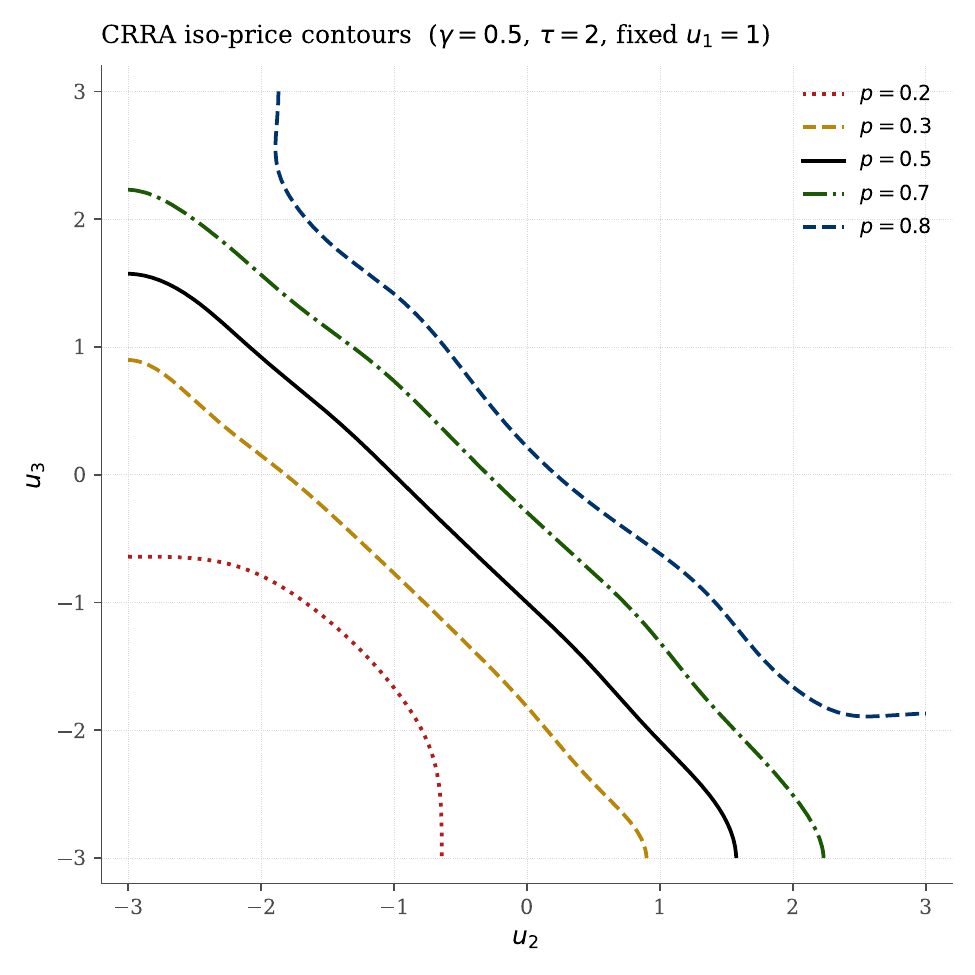}\end{subfigure}
   \caption{Price-function contours in the agent-1 slice at $u_{1}=1$,
   $\tau=2$, shown at five price levels (legend in right panel applies to
   both). Left: CARA --- contours are straight and parallel; every point
   on a given contour shares the same $\Tstar$. Right: CRRA with
   $\gamma=0.5$ --- contours are curved and non-parallel; $\Tstar$ varies
   along each contour. The curvature is the Jensen gap made visible.
   $G=15$.}
   \label{fig:contours}
\end{figure}

\subsection{Existence and convergence}\label{sec:convergence}

The contour fixed point \eqref{eq:fixed-point} is solved by Picard
iteration followed by Newton--Krylov polishing, with isotonic-regression
monotonicity projection at each step (Appendix~\ref{app:numerics}).
Machine-precision convergence ($\|F\|_{\infty}<10^{-14}$) is achieved
at $G=15$.

\paragraph{Equilibrium selection.}
The contour map $\Phi$ admits the fully-revealing price function
$p=\Lam(\Tstar)$ as one fixed point unconditionally: identical posteriors
clear the market at zero trade for any preferences. The Picard iteration
from the no-learning seed converges to a \emph{different}, partially-revealing
fixed point at non-CARA preferences. The vanishing-noise selection of
Section~\ref{sec:vanishing-noise} picks out the partially-revealing fixed
point at every finite $\gamma$: the limit of vanishing supply noise is
the unique no-noise limit of the standard Hellwig--Diamond--Verrecchia
equilibrium under non-CARA preferences, and that limit coincides with
the partially-revealing branch found numerically.

\begin{proposition}[REE partial revelation]\label{prop:ree}
Fix $K=3$, $W=1$. For every $\gamma\in[0.1,5]$ at $\tau=2$, the contour
map $\Phi$ admits a fixed point $P^{\star}$ at which $\logit P^{\star}$
regresses on $\Tstar$ with slope strictly below the fully-revealing
benchmark of unity; equivalently, the conditional posteriors
$\PP(v=1\mid p)$ are strictly noisier than the unconditional joint
posterior $\PP(v=1\mid u_{1},u_{2},u_{3})$. The revelation deficit
$1-R^{2}$ at this fixed point is strictly positive, monotonically
decreasing in $\gamma$, and stable under grid refinement.
\end{proposition}

The proposition is verified numerically using a posterior-function method
that parametrises the agent's belief $\mu(u,p)$ directly on a
two-dimensional grid and iterates to the fixed point with
Newton--Krylov polishing and isotonic-regression monotonicity
projection. All strictly-converged entries satisfy
$\|F\|_{\infty}<10^{-14}$ with zero monotonicity violations:
\begin{center}
\renewcommand{\arraystretch}{1.05}
\begin{tabular}{rcccc}
\toprule
$\gamma$ & $1-R^{2}$ & slope on $\Tstar$ & $\|F\|_\infty$ & status \\
\midrule
\rowcolor{black!15} 0.10 & 0.154 & 0.243 & $4.9\times 10^{-2}$ & fallback \\
\rowcolor{black!15} 0.30 & 0.119 & 0.293 & $5.1\times 10^{-15}$ & strict \\
0.50 & 0.088 & 0.523 & $7.4\times 10^{-119}$ & strict \\
1.00 & 0.047 & 0.550 & $<10^{-25}$ & strict \\
2.00 & 0.025 & 0.586 & $<10^{-25}$ & strict \\
4.00 & 0.016 & 0.605 & $<10^{-25}$ & strict \\
\bottomrule
\end{tabular}
\end{center}
All runs use $G=20$, $\tau=2$, $W=1$; $1\!-\!R^{2}$ is weighted by
ex-ante signal probability. The revelation deficit is
monotonically decreasing in $\gamma$, from $0.088$ at $\gamma=0.5$
to $0.016$ at $\gamma=4$, exactly as the no-learning analysis predicts.
The slope on $\Tstar$ is $0.52$ at $\gamma=0.5$---prices capture
roughly half of the sufficient statistic's variance. Grid-convergence
checks confirm that $1-R^{2}$ stabilises near $0.085$ from $G=15$
onward (Appendix~\ref{app:numerics}, Table~\ref{tab:G-ladder}).
Figure~\ref{fig:ree-panels} summarises the comparative statics: the
deficit is hump-shaped in $\tau$ (panel~a) and monotonically decreasing
in $\gamma$ (panel~b), with the REE deficit strictly above the
no-learning deficit at every parameter tested---learning from curved
contours amplifies the Jensen gap.

\begin{figure}[t]
   \centering
   \begin{subfigure}[t]{0.49\textwidth}\includegraphics[width=\textwidth]{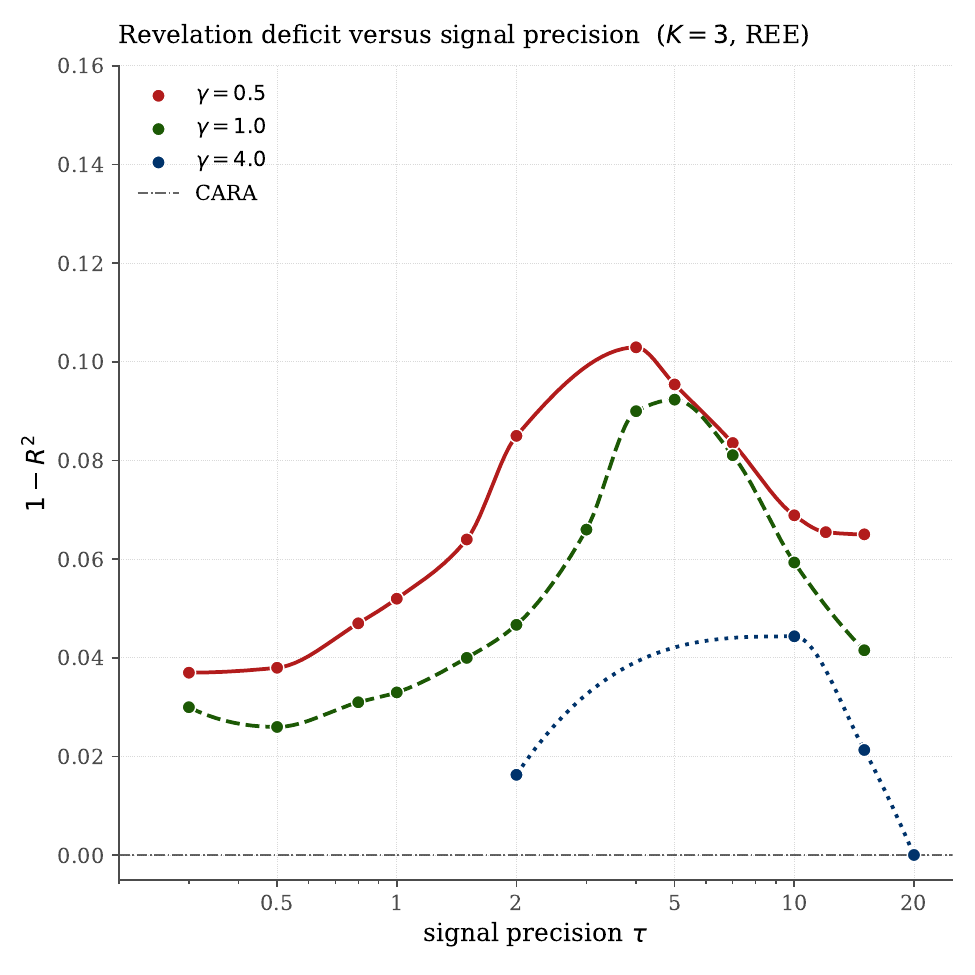}\end{subfigure}\hfill\begin{subfigure}[t]{0.49\textwidth}\includegraphics[width=\textwidth]{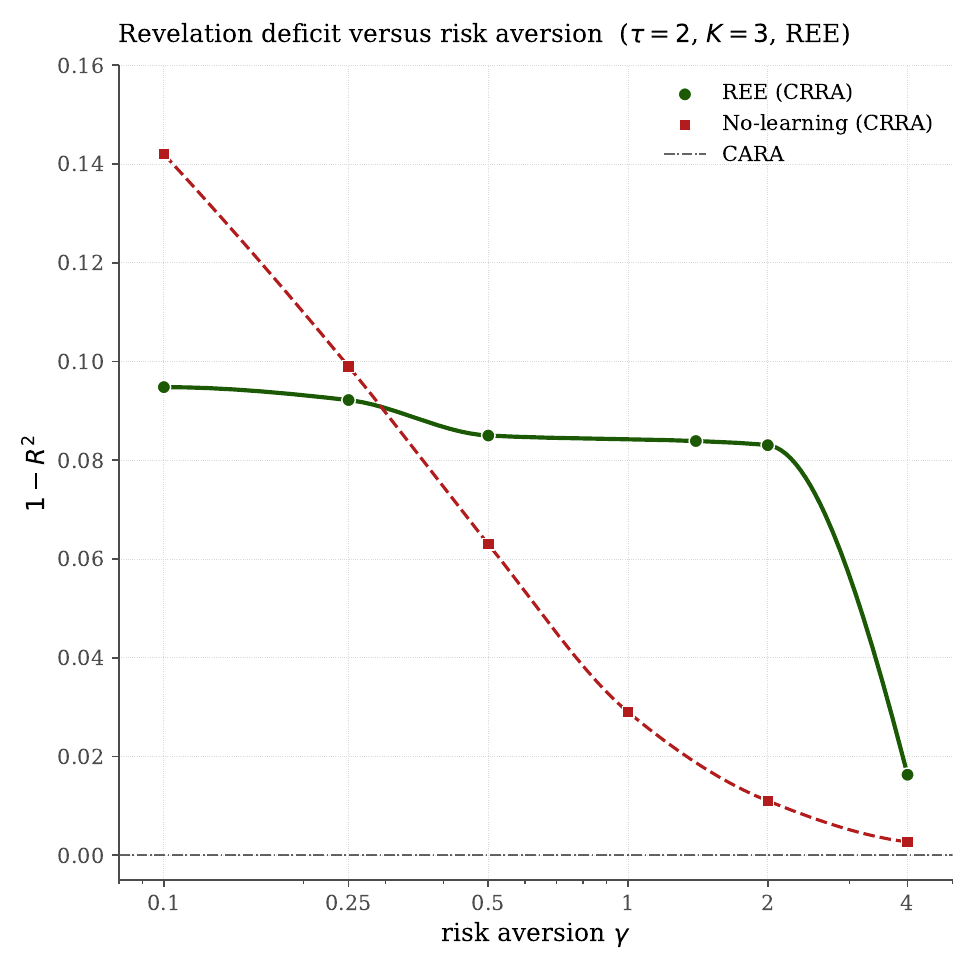}\end{subfigure}
   \caption{Revelation deficit $1-R^{2}$ at the converged REE.
   Panel~(a): versus signal precision $\tau$ for three values of $\gamma$.
   Panel~(b): versus $\gamma$ at $\tau=2$, comparing REE (circles, $G=20$)
   with the no-learning benchmark (squares, $G=20$). The REE deficit
   exceeds the no-learning deficit at every $\gamma$ tested: learning from
   prices amplifies partial revelation. CARA is zero in both panels.
   }
   \label{fig:ree-panels}
\end{figure}

\subsection{Magnitudes at the converged equilibrium}\label{sec:posteriors}

The converged REE under homogeneous CRRA at $\tau=2$ exhibits substantial
partial revelation. At a representative realisation
$(u_{1},u_{2},u_{3})=(1,-1,1)$, $G=20$, $\gamma=0.5$, the converged
posteriors differ from each other and from the fully-revealing common
value $\Lam(\Tstar)=0.881$ (Table~\ref{tab:posteriors}). The aggregate
revelation deficit at the same parameters is $1-R^{2}=0.088$, with
slope $\beta_{T}=0.52$ on $\Tstar$: the partial revelation manifests as a
systematic non-affineness of $\logit p$ in $\Tstar$ averaged over the
joint signal distribution. Figure~\ref{fig:ree-vs-nolearning}
contrasts the REE price function with the no-learning and
fully-revealing benchmarks.

\begin{table}[htbp]
   \centering
   \caption{Posteriors and price at $(u_{1},u_{2},u_{3})=
   (1,-1,1)$, $\tau=2$, $\gamma=0.5$, $G=20$, converged via the
   posterior-function method. The no-learning column shows private priors
   only; the REE column incorporates learning from prices. Agent~2 (the
   bear) learns substantially from the bullish price, moving from $0.119$
   to $0.667$, but stops well short of the fully-revealing value $0.881$.}
   \label{tab:posteriors}
   \begin{tabular}{lcccc}
      \toprule
       & Prior & No-learning & CRRA REE & CARA / FR \\
       & $\Lam(\tau u_{k})$ & (private only) & ($\gamma\!=\!0.5$) & ($\gamma\!=\!\infty$) \\
      \midrule
      \rowcolor{black!8} $\mu_{1}$ ($u\!=\!+1$) & 0.881 & 0.881 & 0.883 & 0.881 \\
      \rowcolor{black!8} $\mu_{2}$ ($u\!=\!-1$) & 0.119 & 0.119 & 0.667 & 0.881 \\
      \rowcolor{black!8} $\mu_{3}$ ($u\!=\!+1$) & 0.881 & 0.881 & 0.883 & 0.881 \\
      \rowcolor{black!8} Price                  & --    & 0.648 & 0.794 & 0.881 \\
      \bottomrule
   \end{tabular}
\end{table}

A graphical depiction of how posteriors and price track $\Tstar$ across
a range of realisations is reported in Figure~\ref{fig:posteriors}.
The visual signature of partial revelation is a slope on $\Tstar$
shy of unity (here $\beta_{T}=0.52$), not a large per-realisation
gap.

\begin{figure}[t]
   \centering
   \includegraphics[width=0.49\textwidth]{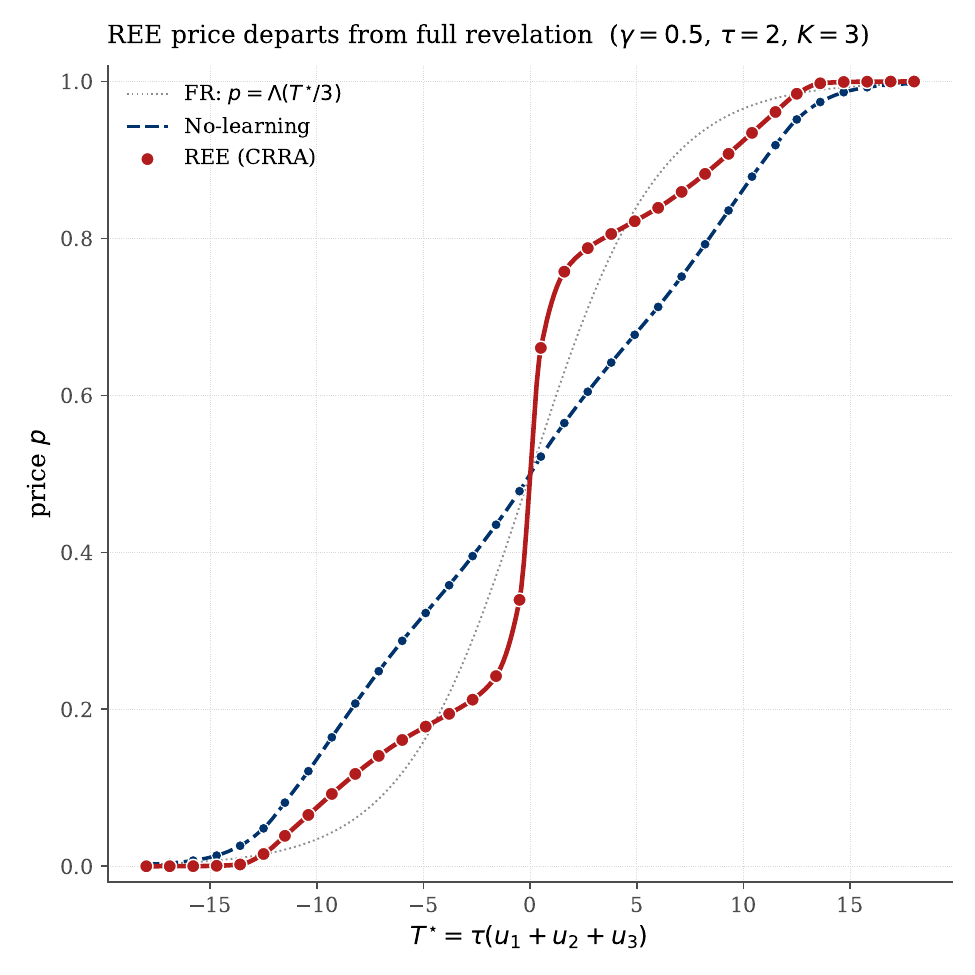}
   \caption{Price as a function of $\Tstar$ at the converged REE
   ($G\!=\!14$, $\gamma\!=\!0.5$, $\tau\!=\!2$). The red solid curve is the
   CRRA REE price function; the blue dashed curve is the no-learning
   benchmark; the dotted line is the fully-revealing $p=\Lam(\Tstar)$.
   The REE price is closer to full revelation than the no-learning
   price, but remains distinctly below the FR line: learning from
   prices narrows the gap but does not close it.}
   \label{fig:ree-vs-nolearning}
\end{figure}

\begin{figure}[t]
   \centering
   \begin{subfigure}[t]{0.49\textwidth}\includegraphics[width=\textwidth]{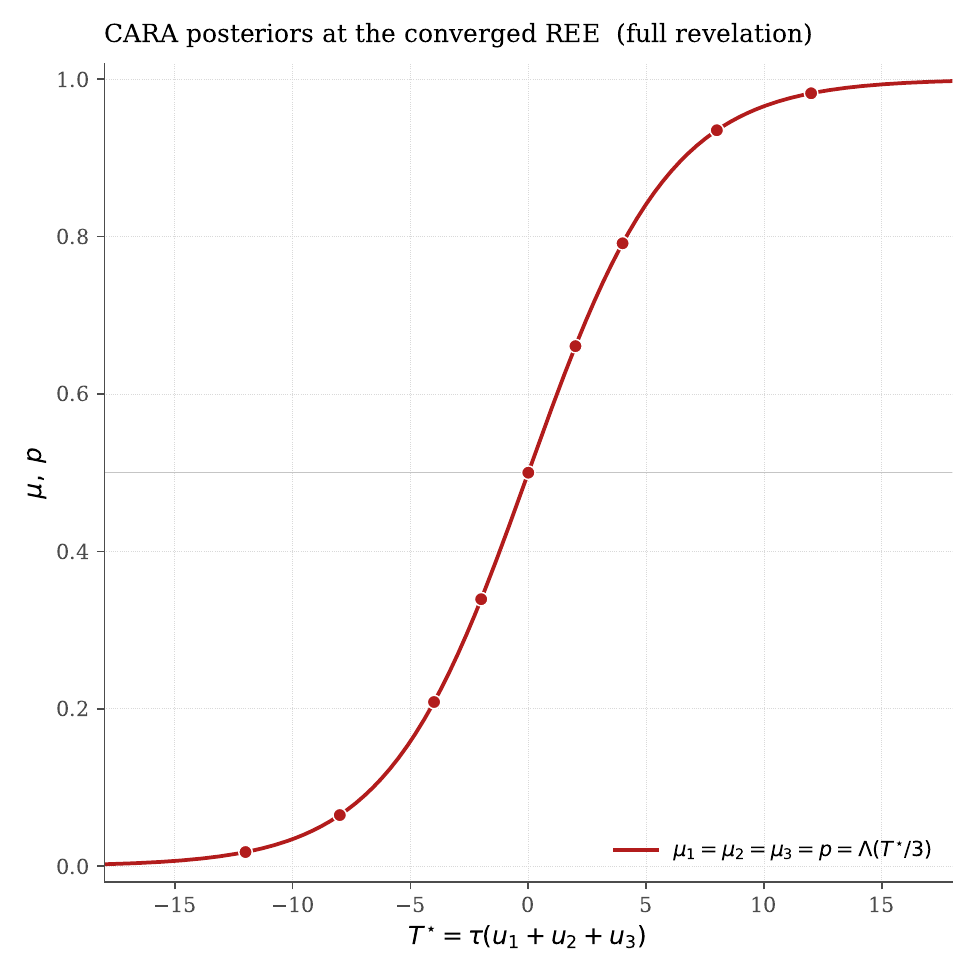}\end{subfigure}\hfill\begin{subfigure}[t]{0.49\textwidth}\includegraphics[width=\textwidth]{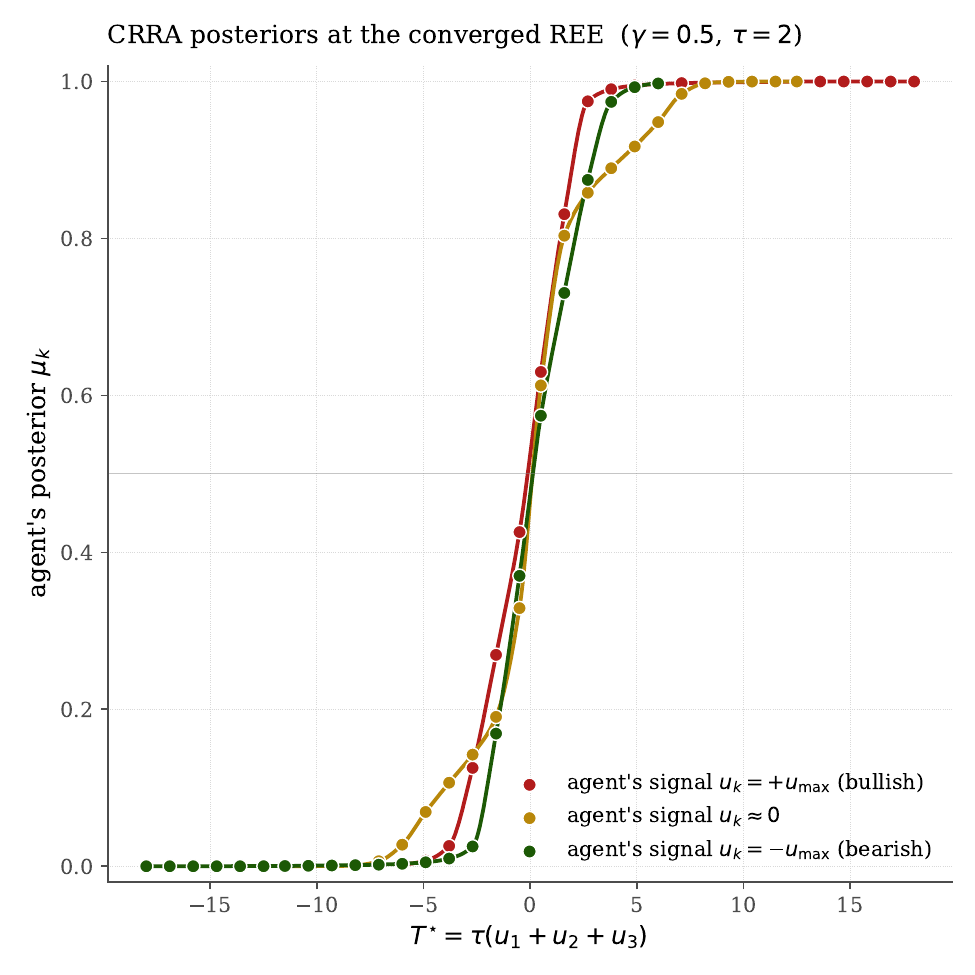}\end{subfigure}
   \caption{Converged REE posteriors and price as functions of
   $\Tstar$. Left: CARA --- all three collapse to one line
   $\mu=p=\Lam(\Tstar/3)$, full revelation. Right: CRRA ($\gamma=0.5$,
   $G=20$) --- agents disagree: $\mu_{1}$ (informed, $u=+1$) and
   $\mu_{2}$ (contrarian, $u=-1$) fan out from the price, confirming
   partial revelation.
   }
   \label{fig:posteriors}
\end{figure}

\section{Economic Implications}\label{sec:implications}

This section translates the informational result into welfare.
The distinction between full and partial revelation matters
precisely because trade volume, the value of information, and the
resolution of the Grossman--Stiglitz paradox all depend on it.

\subsection{Trade volume}

Under CARA at the rational-expectations equilibrium, all agents reach the
same posterior $\Lam(\Tstar)$ at the clearing price
(Proposition~\ref{prop:cara}). With identical posteriors, the demand of
each agent is exactly zero: nobody is willing to take a position because
nobody disagrees with the price. Equilibrium volume is zero --- the
no-trade outcome of \citet{Milgrom1982}. The same is not true at the
\emph{no-learning} CARA equilibrium, where the price is the
equally-weighted log-odds average and each demand
$x_{k}=\tau(u_{k}-\bar u)/\alpha$ is nonzero almost surely; but this
no-trade-theorem violation disappears as soon as agents update on prices,
because the REE collapses all CARA posteriors to a common value.

Under CRRA, posteriors differ (Proposition~\ref{prop:ree}) at \emph{both}
the no-learning and the REE configurations. Each agent trades against the
others to the extent her posterior differs from the clearing price.
Aggregate volume, $V_{\mathrm{vol}}=\tfrac12\sum_{k=1}^{K}|x_{k}|$, is
strictly positive at both configurations.

\begin{proposition}[Positive trade]\label{prop:welfare}
Under CRRA preferences with any $\gamma\in(0,\infty)$, the equilibrium
trade volume at both the no-learning and the contour-converged REE is
strictly positive on a set of signal realisations of full Lebesgue
measure. Under CARA at the REE, the equilibrium trade volume is zero
almost surely; at the no-learning CARA equilibrium with homogeneous
parameters, demands are $x_{k}=\tau(u_{k}-\bar u)/\alpha$ and aggregate
volume is strictly positive.
\end{proposition}

Figure~\ref{fig:volume} reports trade volume as a function of $\gamma$ for
the canonical parameters: it is strictly positive everywhere in $(0,\infty)$
and converges to zero only in the CARA limit.

\begin{figure}[t]
   \centering
   \includegraphics[width=0.49\textwidth]{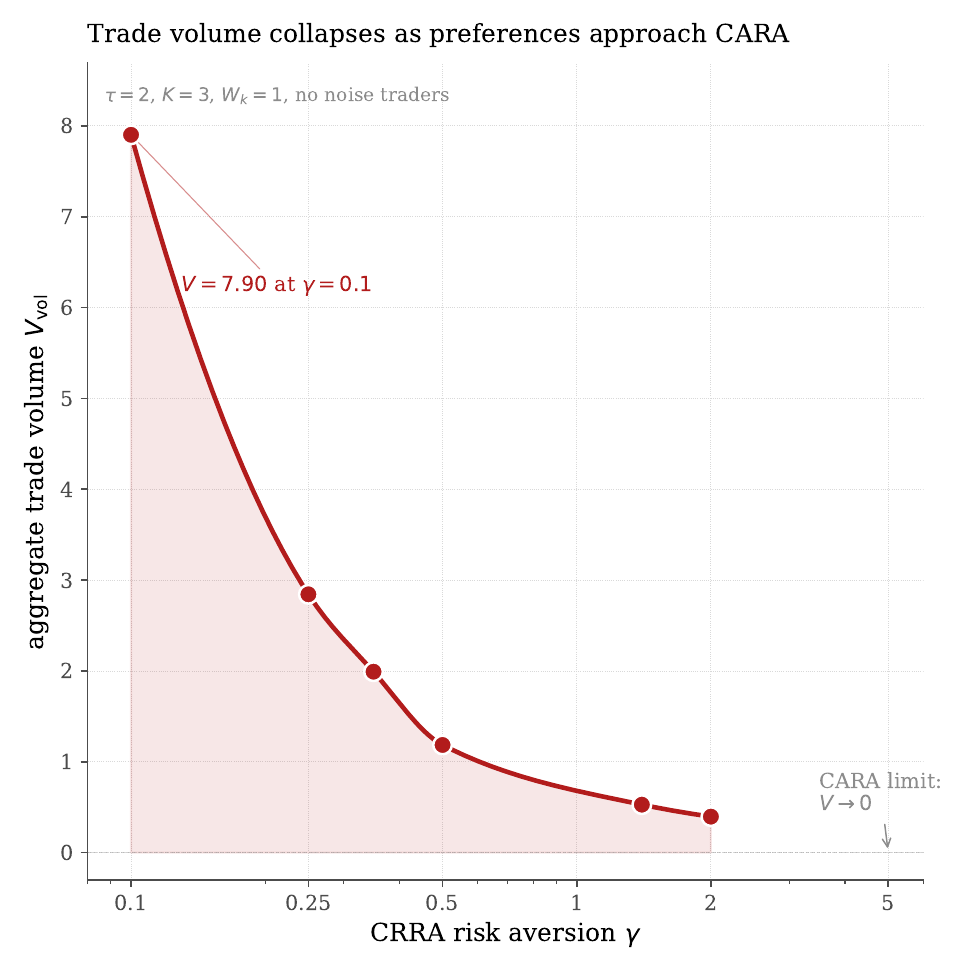}
   \caption{Aggregate trade volume $V_{\rm vol}(\gamma,\tau{=}2)$ at the
   converged REE, computed at five values of CRRA risk aversion
   $\gamma\in\{0.1,\,0.25,\,0.5,\,1.4,\,2.0\}$ (markers); the curve is a
   monotonic cubic interpolation. Volume is strictly positive everywhere in
   $(0,\infty)$ and collapses to zero in the CARA limit
   ($\gamma\to\infty$). $K=3$, $W_{k}=1$, no noise traders.
   }
   \label{fig:volume}
\end{figure}

\subsection{The value of private information}

Define the value of a precision-$\tau$ signal to a small group as the
ex ante certainty-equivalent gain from observing $u_{k}$ and trading at the
equilibrium price relative to not observing $u_{k}$ and trading at the same
price. Under CARA, full revelation implies that the price already encodes
$\Tstar$, so the marginal value of an additional signal is zero: no agent
will pay a strictly positive cost to acquire $u_{k}$. Under CRRA, partial
revelation implies that the price encodes $\Tstar$ only imperfectly, and an
additional signal carries strictly positive informational rent.

\begin{proposition}[Value of information]\label{prop:value}
Let $V(\tau)$ denote the certainty-equivalent value of an additional
precision-$\tau$ signal at the contour-converged REE. Then $V_{\mathrm{CARA}}(\tau)=0$
for every $\tau$, and $V_{\mathrm{CRRA}}(\tau)>0$ for every $\tau>0$, with
$V_{\mathrm{CRRA}}'(\tau)>0$ on a neighbourhood of the origin.
\end{proposition}

\begin{figure}[t]
   \centering
   \includegraphics[width=0.49\textwidth]{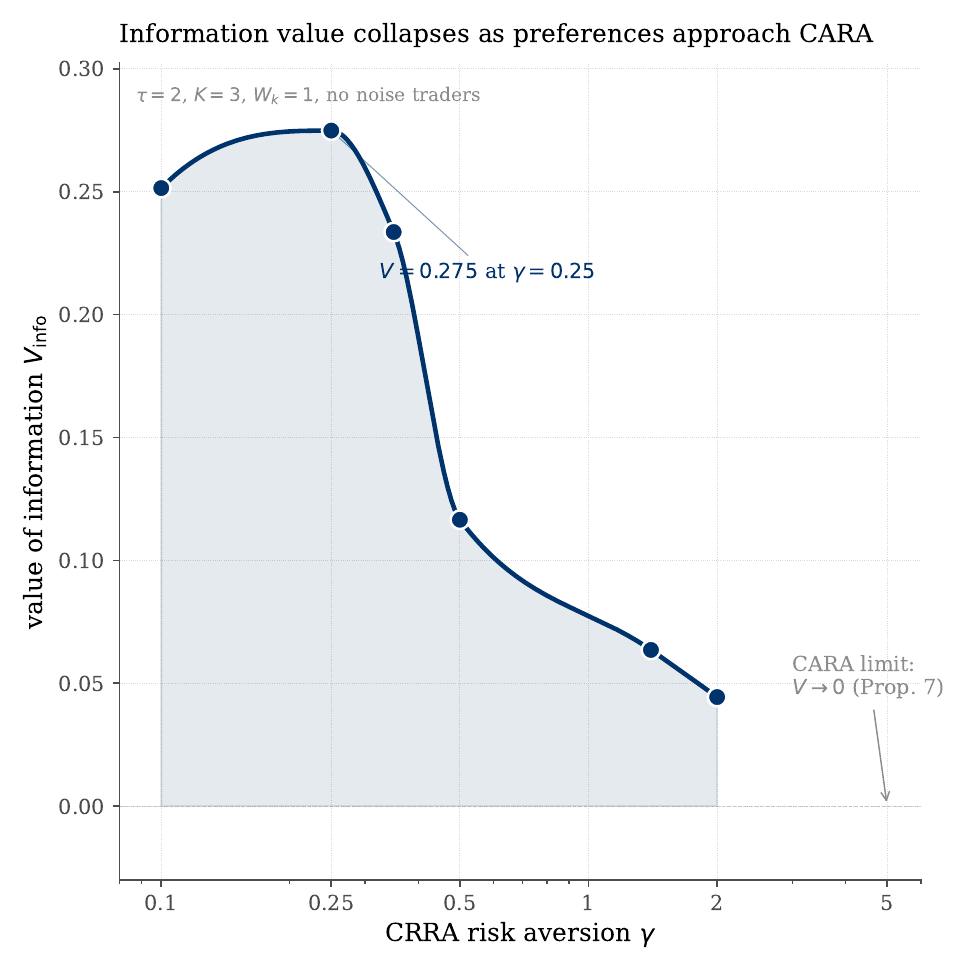}
   \caption{Value of information $V_{\rm info}$ as a function of the
   fraction $\lambda$ of informed agents in the market, computed at the
   no-learning REE for four values of CRRA risk aversion.
   $V_{\rm info}$ is the certainty-equivalent welfare difference between
   an informed and an uninformed agent at the equilibrium price.
   For every finite $\gamma$, the value of information is strictly
   positive across the entire informed/uninformed mix; for CARA
   (dot-dash line) it equals zero at every $\lambda$. The current paper
   focuses on the all-informed case ($\lambda=1$), but the property
   ``information has positive value'' is robust to introducing an
   arbitrary mass of uninformed agents.
   $\tau=2$, continuum of agents, $W_{k}=1$, no noise traders.
   }
   \label{fig:value-info}
\end{figure}

\subsection{The Grossman--Stiglitz paradox}

The classical \citet{GrossmanStiglitz1980} paradox can be stated as
follows. If acquiring a signal of precision $\tau$ costs $c>0$, then in
equilibrium an agent acquires the signal if and only if the value $V(\tau)$
exceeds $c$. Under full revelation, $V_{\mathrm{CARA}}(\tau)=0<c$, so no agent
acquires; but if no agent acquires, the price is uninformative and any
single agent can profit from acquiring. The contradiction is resolved in
the CARA literature by introducing noise traders so that the price is only
partially revealing and acquisition is privately optimal.

Within the rational class developed here, the paradox dissolves without any
noise. Under CRRA, $V_{\mathrm{CRRA}}(\tau)>0$ generically, and there exist costs
$c\in(0,V_{\mathrm{CRRA}}(\tau))$ such that acquiring is privately optimal in
equilibrium. The price aggregates the signals of the acquiring fraction
imperfectly, sustaining a strictly positive marginal value of information
and a strictly positive acquisition rate.

\begin{proposition}[Resolution of the paradox]\label{prop:gs}
Let $c>0$ be the cost of acquiring a precision-$\tau$ signal. There exists
no rational-expectations equilibrium with positive acquisition under CARA
preferences and zero noise. Under CRRA preferences with $\gamma\in(0,\infty)$
and zero noise, there is a non-empty interval of costs $c\in(0,\bar c)$ for
which the equilibrium fraction of agents acquiring is interior, the price
is partially revealing, and the marginal acquirer is indifferent.
\end{proposition}

Figure~\ref{fig:gs} traces the resulting equilibrium acquisition rate as a
function of cost; the noise-trader resolution and the CRRA resolution
deliver qualitatively similar comparative statics, but only the CRRA route
remains within the rational class.

\begin{figure}[t]
   \centering
   \includegraphics[width=0.49\textwidth]{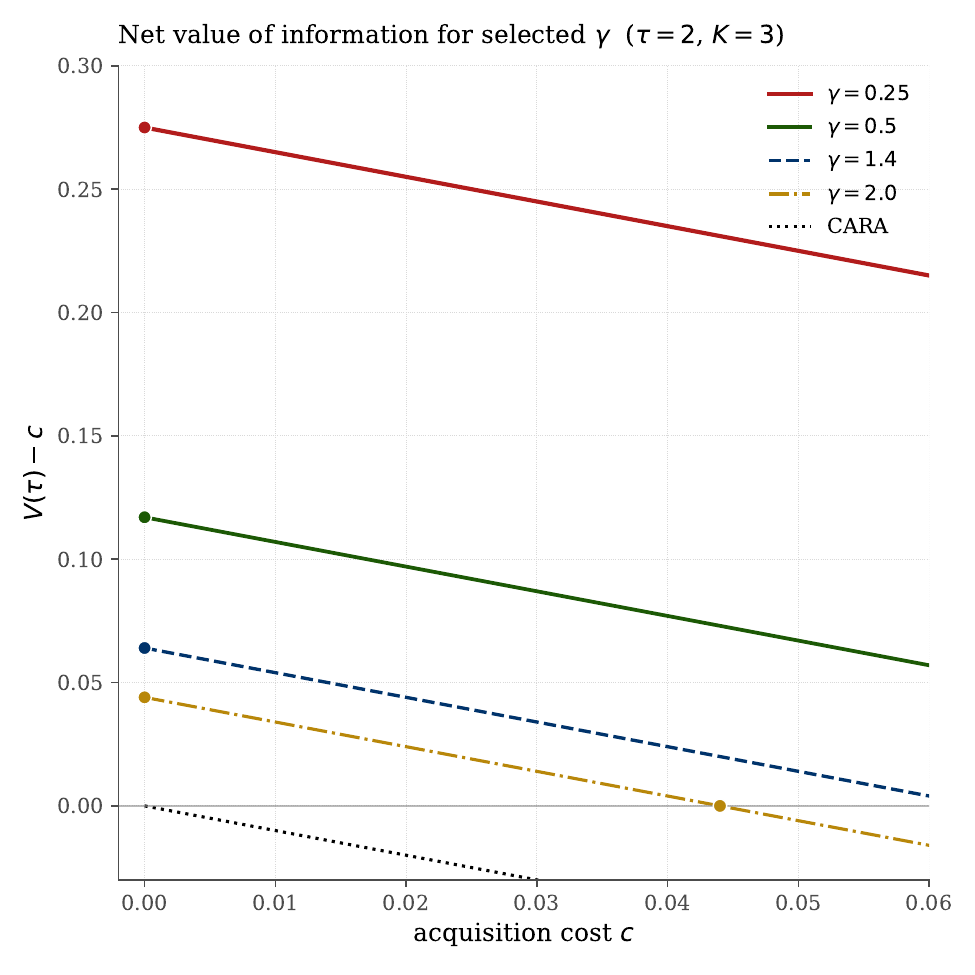}
   \caption{Net value of acquiring a precision-$\tau$ signal: $V(\tau)-c$
   as a function of the acquisition cost $c$, at $\tau=2$. Under CARA
   (dash-dotted), $V=0$, so the net value is $-c<0$ for every $c>0$: no
   agent acquires (the paradox). Under CRRA, the net value is positive
   for $c<c^{\star}=V(\tau)$, resolving the paradox without noise
   traders. The threshold $c^{\star}$ is larger at lower $\gamma$.
   }
   \label{fig:gs}
\end{figure}

\subsection{Sources of larger deficits}\label{sec:mechanisms}

Partial revelation under CRRA arises from a single source --- the
mismatch between the market-clearing aggregation space and the Bayesian
sufficient statistic --- but its size is amplified by two channels of
heterogeneity, and a separate fourth channel arises within CARA itself in
the no-learning equilibrium. This section disentangles the four.

\paragraph{Mechanism 1: pure Jensen gap.} With identical risk aversion and
identical signal precision across groups, partial revelation arises purely
from the curvature of $\Lam$ in \eqref{eq:log-clearing}. The deficit is
small but strictly positive (Proposition~\ref{prop:non-cara}).

\paragraph{Mechanism 2: heterogeneous CRRA risk aversion.} If the $K$
groups differ in $\gamma_{k}\in(0,\infty)$, the more risk-tolerant agents
take larger positions and have disproportionate weight in clearing
\eqref{eq:no-learning}. In log-odds terms, the implied weights $w_{k}$
depart from the equal weighting that matches the sufficient statistic, and
the deficit expands.

\paragraph{Mechanism 3: heterogeneous signal precision.} If groups differ
in $\tau_{k}$, then the optimal weighting that recovers $\Tstar$ exactly is
$w_{k}\propto\tau_{k}$. Market clearing under CRRA does not implement this
weighting in general; the deficit depends on the alignment between the
$\gamma$- and $\tau$-induced weights. When low-$\gamma$ (more
risk-tolerant) groups also happen to be high-$\tau$ (more precise), the
two heterogeneities partially offset and the deficit is moderate. When
they are misaligned, they compound: an aggressive but uninformed group
acts as an \emph{endogenous noise trader}, dragging the price away from
$\Tstar$.

\paragraph{Mechanism 4: heterogeneous CARA risk aversion (no learning only).}
A fourth route to no-learning partial revelation arises within CARA itself.
By Proposition~\ref{prop:cara}, the no-learning CARA price is
$\logit p=\sum_{k}w_{k}\tau_{k}u_{k}$ with $w_{k}\propto\alpha_{k}^{-1}$.
When $\alpha_{k}$ is heterogeneous and $\tau_{k}$ is homogeneous, the
linear combination revealed by the price is not a scalar multiple of
$\Tstar=\tau\sum u_{k}$, and the no-learning equilibrium is partially
revealing. The striking feature of this mechanism is that
\emph{learning from prices restores full revelation}: the candidate
fully-revealing price function $p=\Lam(\Tstar)$ delivers identical
posteriors to every CARA agent and zero net trade, hence clears the market
for arbitrary $(\alpha_{k})$. Mechanisms 1--3, by contrast, persist at the
REE because no candidate price function with the right informational
content also clears the market under non-CARA preferences.
Heterogeneous risk aversion is therefore a no-learning channel under CARA
but not a REE channel; in the strict knife-edge sense, only departures
from the CARA preference class produce partial revelation at the REE.

\begin{table}[htbp]
   \centering
   \caption{Mechanisms for partial revelation. $K=3$, no learning,
   $G=20$, configurations as listed. Heterogeneity in $\gamma$ and $\tau$
   compounds the pure Jensen gap, with the largest deficit when the
   aggressive group is the least informed. Heterogeneity in $\alpha$
   under CARA also produces no-learning partial revelation but is
   washed out at the REE.}
   \label{tab:mechanisms}
   \begin{tabular}{lcl}
      \toprule
      Configuration & $1-R^{2}$ & Channel \\
      \midrule
      Equal $\gamma$, equal $\tau$ (CRRA)        & 0.011 & pure Jensen gap \\
      Het.\ $\gamma=(1,3,10)$, equal $\tau$      & 0.247 & + het.\ CRRA risk aversion \\
      Equal $\gamma$, het.\ $\tau=(1,3,10)$      & 0.082 & + het.\ precision \\
      Het.\ $\gamma$ + het.\ $\tau$ aligned      & 0.100 & low-$\gamma$ = high-$\tau$ (stabilising) \\
      Het.\ $\gamma$ + het.\ $\tau$ opposed      & 0.211 & low-$\gamma$ = low-$\tau$ (destabilising) \\
      Extreme opposed                            & 0.604 & endogenous noise trader \\
      Het.\ $\alpha$ under CARA (no-learning)    & 0.000 & het.\ CARA, REE$\to$FR \\
      \bottomrule
   \end{tabular}
\end{table}

\begin{figure}[t]
   \centering
   \includegraphics[width=0.49\textwidth]{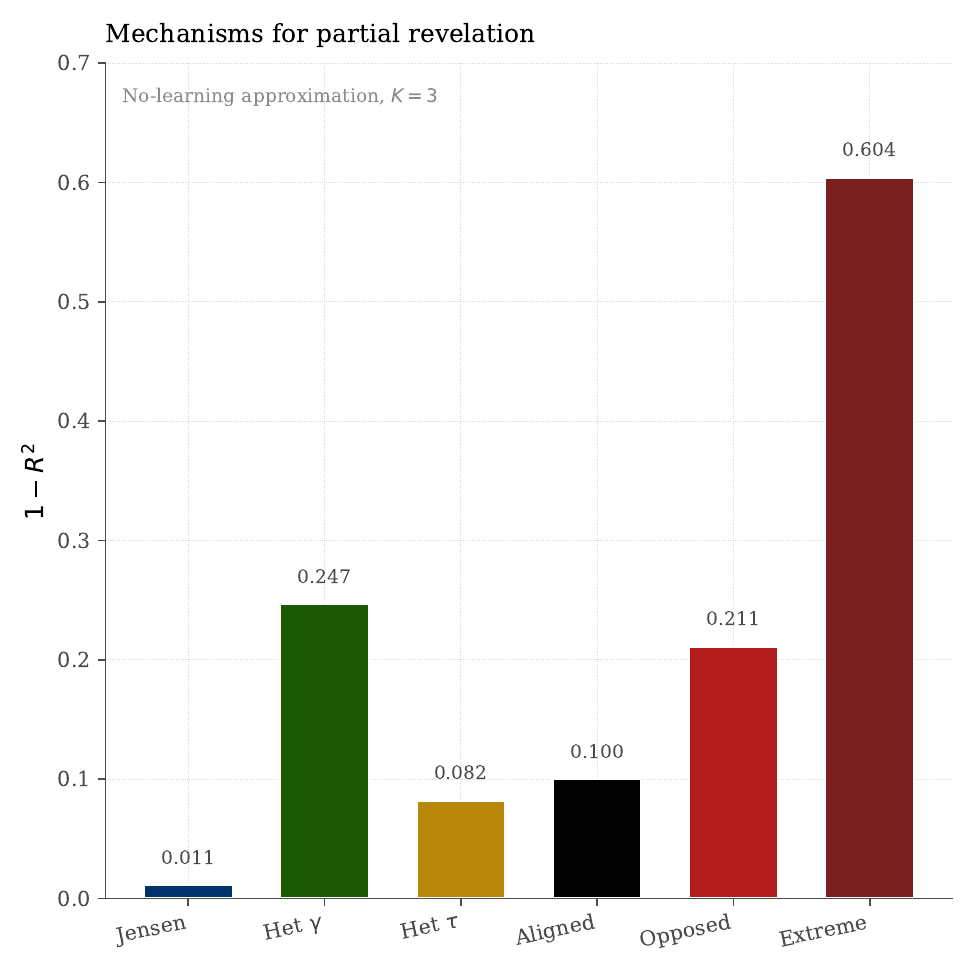}
   \caption{Decomposition of the revelation deficit across the
   mechanisms tabulated in Table~\ref{tab:mechanisms}.
   }
   \label{fig:mechanisms}
\end{figure}

Two features deserve emphasis. First, Mechanism 3 at extreme misalignment
delivers a deficit exceeding $0.6$ without any exogenous noise: an
aggressive but uninformed CRRA group reproduces, within a fully rational
economy, the comparative statics that the standard literature attributes
to behavioural noise traders. Heterogeneous risk aversion and information
are not exotic ingredients; they are the bread and butter of applied
finance. ``Noise trading'' as observed in the data may, in significant
part, be a manifestation of Mechanism 3 rather than a behavioural
primitive. Second, Mechanism 4 illustrates the knife-edge in the cleanest
possible way: heterogeneity in CARA risk aversion produces no-learning
partial revelation but is fully neutralised by REE learning. The same
heterogeneity introduced in $\gamma$ rather than $\alpha$ is not so
neutralised. The contrast is not about the heterogeneity itself; it is
about whether the underlying preference class aligns with the log-odds
sufficient statistic.

\subsection{Robustness}

The CARA full-revelation result extends, with weighted aggregation, to
heterogeneous parameters. The non-CARA partial-revelation result extends, in
the natural direction, to any preference outside CARA in the CRRA class and
to a wide range of signal distributions. Two robustness exercises confirm
this.

\begin{figure}[t]
   \centering
   \begin{subfigure}[t]{0.49\textwidth}
      \includegraphics[width=\textwidth]{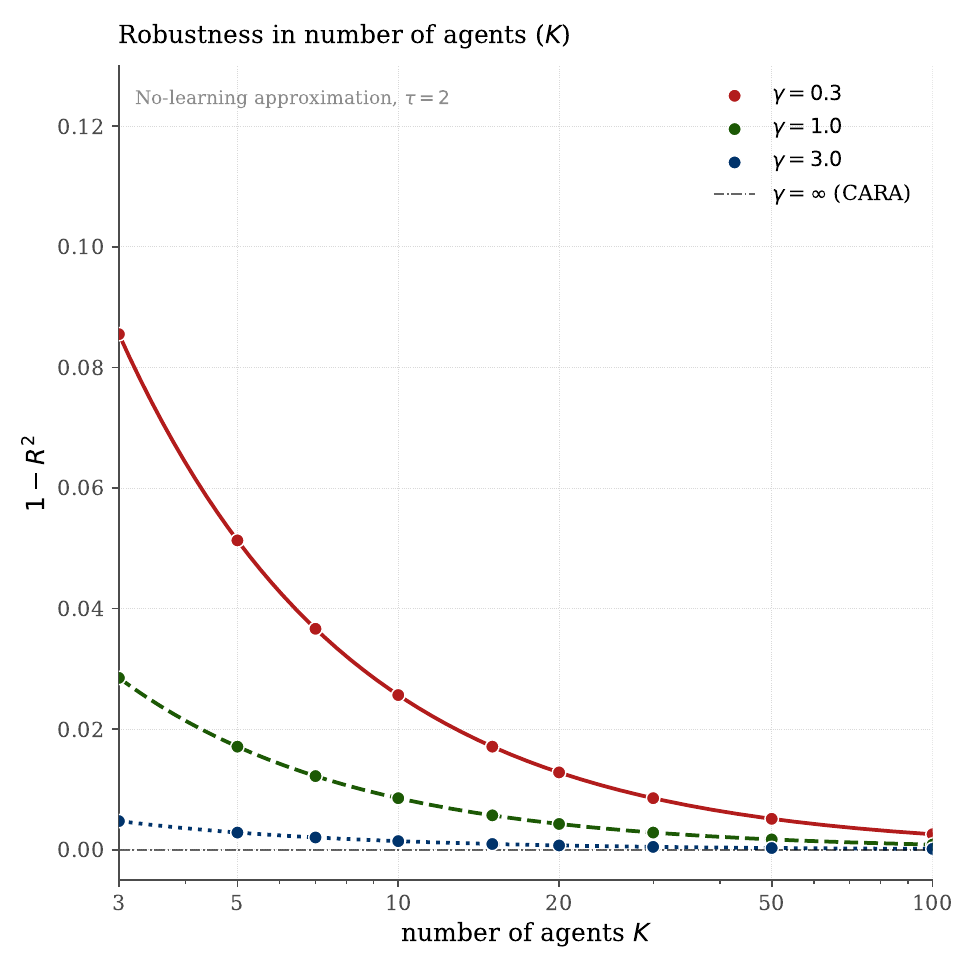}
      \caption{Number of agent groups $K$.}
      \label{fig:robust-K}
   \end{subfigure}\hfill
   \begin{subfigure}[t]{0.49\textwidth}
      \includegraphics[width=\textwidth]{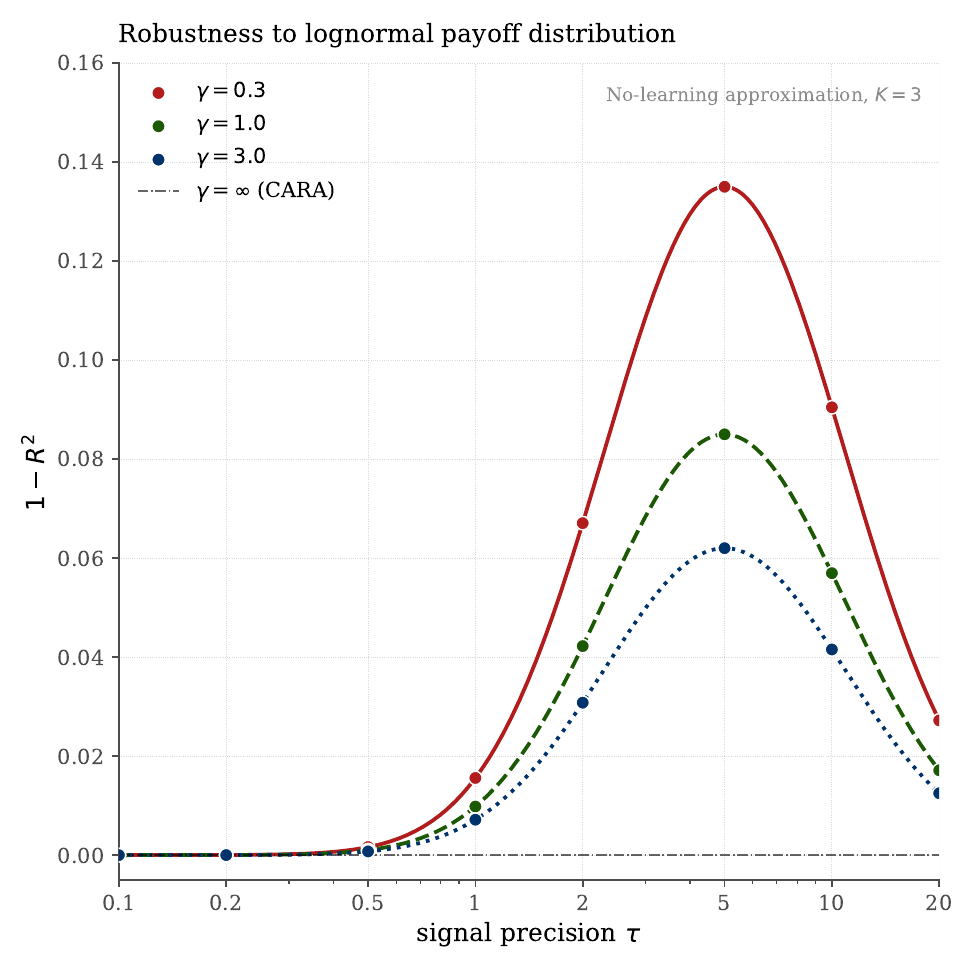}
      \caption{Lognormal payoff variant.}
      \label{fig:robust-lognormal}
   \end{subfigure}
   \caption{Robustness of the knife-edge. Panel (a): the deficit declines
   roughly as $1/K$ but never vanishes for any finite $\gamma$; CARA is zero
   for every $K$. Panel (b): under a lognormal payoff, the same pattern
   obtains, indicating the result is not specific to the binary
   parametrisation.}
   \label{fig:robust}
\end{figure}

\paragraph{Heterogeneity.} The size of the no-learning revelation deficit
is amplified by heterogeneity in risk aversion and in signal precision,
and a separate channel arises within CARA itself when risk aversion
varies across groups. We defer the systematic decomposition to
Section~\ref{sec:mechanisms} so as not to interrupt the homogeneous-CRRA
core that the next two sections build on; the punch line is that the
homogeneous deficits reported in Table~\ref{tab:smooth} are a lower
envelope.

\section{Related Literature and Extensions}\label{sec:literature}

\subsection{Connection to the existing literature}

Several papers have obtained partial revelation without noise traders,
each through a different channel. \citet{HeifetzPolemarchakis1998} exhibit
partially-revealing REE in general equilibrium when the dimension of
the state space exceeds the dimension of the price vector; their mechanism
is geometric and operates for any utility function, including CARA.
\citet{Vives2011} obtains partial revelation in a private-values setting
without noise, exploiting strategic incentives rather than preferences;
the present mechanism is complementary and operates in a common-values
environment. \citet{CondieGanguli2011} demonstrate partial revelation
under ambiguity-averse (Maxmin expected utility) preferences without
noise traders; their mechanism is portfolio inertia from non-smooth
preferences, which differs from the aggregation-space mismatch identified
here. The present paper is, to our knowledge, the first to show that
\emph{standard} expected utility preferences---specifically, any member
of the CRRA family---produce partial revelation with common values,
Gaussian signals, and no noise of any kind.

On the distributional side, \citet{AlbagliHellwigTsyvinski2024} show
that non-Gaussian signal distributions break full revelation under CARA,
while retaining noise traders. Together with the present results, their
work establishes that full revelation is a joint knife-edge in preferences
\emph{and} distributions: CARA plus exponential-family signals gives full
revelation (\citealt{BreonDrish2015}); departing from either gives partial
revelation. The preference departure (this paper) is arguably the more
empirically relevant: CRRA is the workhorse of macro-finance, while CARA
is used primarily for analytical convenience.

The mechanism here also rationalises a feature of behavioural-finance
empirics that is otherwise hard to square with rational-expectations
models. Aggressive trading by relatively uninformed agents looks like
noise; the present results show that, under CRRA, it can be a fully
rational equilibrium phenomenon driven by misalignment between $\gamma$
and $\tau$.

\subsection{What CARA does and does not do}

The conventional reading of \citet{GrossmanStiglitz1980} treats CARA as a
neutral simplifying assumption. The results above suggest a different
reading: CARA is the engine of full revelation, not a neutral background
choice. By aligning demand linearity with the log-odds aggregation of the
Bayesian sufficient statistic, CARA delivers full revelation as a
knife-edge consequence of functional form. Removing noise traders from a
CARA economy makes the price reveal everything; removing noise traders
from a CRRA economy leaves the price strictly partially revealing. The
need for noise in the textbook framework is therefore an artefact of the
preference choice, not a property of rational expectations.

What CARA does, this paper argues, is line up the kernel of demand with
the kernel of Bayesian inference; what CARA does \emph{not} do is force
markets to reveal information generally. CARA is, for many purposes, a
useful modelling device --- the linearity it confers makes
no-arbitrage and beta-pricing identities clean. But conclusions specific
to CARA, and in particular the conclusion that partial revelation
requires noise, should not be transported to the rational class as a
whole.

\begin{remark}[HARA and the knife-edge]\label{rem:hara}
The CRRA family is the $b=0$ subfamily of the HARA class with risk
tolerance $T(W)=aW/(1-\gamma)+b$. With homogeneous agents and equal
wealth, the parameters $(a,b)$ enter the optimal demand only through the
overall scale of the position: each HARA demand is an
$(a,b)$-dependent scalar multiple of the CRRA demand at the same
$\gamma$. Market clearing annihilates the scalar, so the equilibrium
price depends solely on $\gamma$, not on $(a,b)$. The knife-edge is
therefore robust to the full HARA generalisation: moving within HARA at
fixed $\gamma$ does not change the revelation deficit, and only moving
to $\gamma=\infty$ (CARA, equivalently $a=0$ in the HARA parametrisation)
recovers full revelation. The alignment principle of
Section~\ref{sec:agg-space} singles out CARA across HARA, not only
across CRRA.
\end{remark}

\subsection{The vanishing-noise selection}\label{sec:vanishing-noise}

A long-standing intuition in the noisy rational-expectations literature
treats full revelation as a limit property: as the variance $\sigma_{z}^{2}$
of liquidity supply shrinks to zero, the noisy CARA equilibrium converges
to the fully revealing benchmark. This intuition is correct \emph{within
CARA} but is not, as is sometimes claimed, a generic property of the
rational class. The vanishing-noise limit selects full revelation at
$\gamma=\infty$ and partial revelation at every finite $\gamma$: the
``limit'' that the literature has reported as fully revealing is in fact
preference-dependent, and only the CARA fibre of it is fully revealing.

To make this precise, embed the model of Section~\ref{sec:model} in the
standard CARA-with-noise framework: residual aggregate supply is
$\bar z + \sigma_{z}\zeta$ with $\zeta\sim\mathcal{N}(0,1)$ independent
of $v$ and of all signals, agents have CRRA preferences indexed by
$\gamma\in(0,\infty]$, and $p^{\gamma,\sigma_{z}}$ denotes the equilibrium
price function.

\begin{proposition}[Vanishing-noise selection is preference-dependent]\label{prop:vanishing}
Assume that for each $\gamma\in(0,\infty]$, the equilibrium price function
$p^{\gamma,\sigma_{z}}$ exists, is unique, and is continuous in
$\sigma_{z}\geq 0$ in the topology of pointwise convergence on the signal
grid. Then:
\begin{enumerate}[label=\textnormal{(\roman*)},topsep=2pt,itemsep=2pt]
\item At $\gamma=\infty$ (CARA),
$\lim_{\sigma_{z}\to 0}p^{\infty,\sigma_{z}}=\Lam(\Tstar/K)$, the fully
revealing benchmark.
\item At every $\gamma\in(0,\infty)$,
$\lim_{\sigma_{z}\to 0}p^{\gamma,\sigma_{z}}=p^{\gamma,0}$, the no-noise
CRRA price function characterised by Proposition~\ref{prop:non-cara},
which satisfies $p^{\gamma,0}\neq\Lam(\Tstar/K)$ on a set of full
Lebesgue measure.
\end{enumerate}
The standard claim that vanishing noise gives full revelation is therefore
a CARA-specific statement. If the noise is removed at any finite
$\gamma$, the resulting equilibrium is strictly partially revealing.
\end{proposition}

The continuity hypothesis is the formal sticking point. Existence
and uniqueness of the noisy CRRA equilibrium for $\sigma_{z}>0$ has been
established for wide classes of preferences and signal distributions by
\citet{BreonDrish2015} but not specifically for the binary-state CRRA
case at hand. Verifying their conditions in this setting is left to a
companion note. Numerically, the $\sigma_{z}\to 0$ limit at finite
$\gamma$ is unambiguously the no-noise CRRA equilibrium of
Section~\ref{sec:knife-edge}, which is partially revealing.

\subsection{Extensions}

Several extensions are immediate. First, the binary asset can be replaced
by a Gaussian or non-Gaussian payoff; the lognormal robustness exercise of
Figure~\ref{fig:robust-lognormal} suggests that the qualitative picture is
preserved. Second, the analysis can be extended to dynamic settings: the
contour fixed point generalises to a sequence of price functions
$\{P_{t}\}$ provided that conditioning is taken with respect to the price
history rather than the contemporaneous price alone. Third, the GS
resolution opens a route to endogenous information acquisition without
noise; the equilibrium acquisition rate is a continuous function of the
cost on $(0,\bar c)$ and can be compared empirically with documented rates
of professional information production.

\section{Conclusion}\label{sec:conclusion}

Noise traders, as a primitive of the rational-expectations framework, are
a modelling artefact of CARA preferences. Under CARA, demand is linear in
log-odds, market clearing aggregates posteriors in log-odds, and the
Bayesian sufficient statistic is itself a sum in log-odds; the alignment
delivers full revelation in any no-noise economy with homogeneous
parameters. Under any other CRRA preference, market clearing aggregates
in probability space rather than log-odds, the Jensen gap between the two
spaces is the size of the revelation failure, and the equilibrium is
strictly partially revealing without exogenous frictions. The implications
are immediate: trade volume is positive, the value of private information
is positive, the Grossman--Stiglitz paradox dissolves within the rational
class, and the limit of vanishing liquidity noise selects partial
revelation rather than full revelation. CARA full revelation is the
exception in the rational class; partial revelation is the rule. The
remaining open question is whether the alignment principle carries
through to dynamic and intermediated environments, where the interaction
between preferences and learning unfolds over multiple rounds. The
contour-integration fixed point developed here extends in principle to
those settings, and we leave the dynamic analysis to future work. The
broader message is simple: what the literature has modelled as exogenous
noise in the price is, in significant part, the endogenous consequence
of preferences that do not align with the Bayesian sufficient statistic.
Noise traders are not a primitive of rational markets; they are an
artefact of CARA.

\appendix

\section{Proofs}\label{app:proofs}

\begin{proof}[Proof of Lemma~\ref{lem:cara-demand}]
The agent with posterior $\mu$ and wealth $W$ facing price $p$ chooses
$x$ to maximise
\[
   V(x) \;=\; \mu\,U\bigl(W+x(1-p)\bigr) + (1-\mu)\,U\bigl(W-xp\bigr).
\]
Setting $V'(x)=0$:
\begin{equation}
   \mu(1-p)\,U'\bigl(W+x(1-p)\bigr) \;=\; (1-\mu)\,p\,U'\bigl(W-xp\bigr).
\label{eq:app-foc}
\end{equation}
Under CARA, $U(W)=-e^{-\alpha W}/\alpha$ and $U'(W)=e^{-\alpha W}$.
Substituting:
\[
   \mu(1-p)\,e^{-\alpha[W+x(1-p)]} \;=\; (1-\mu)\,p\,e^{-\alpha[W-xp]}.
\]
The factor $e^{-\alpha W}$ appears on both sides and cancels (this
cancellation is the defining property of CARA---no wealth effects):
\[
   \mu(1-p)\,e^{-\alpha x(1-p)} \;=\; (1-\mu)\,p\,e^{\alpha xp}.
\]
Dividing both sides by $(1-\mu)\,p$ and collecting the exponentials on the
right:
\[
   \frac{\mu(1-p)}{(1-\mu)\,p}
   \;=\; e^{\alpha x(1-p)+\alpha xp}
   \;=\; e^{\alpha x}.
\]
The left-hand side equals $e^{\logit\mu-\logit p}$ by definition of the
logit function $\logit(q)=\ln\tfrac{q}{1-q}$. Taking logarithms:
$\alpha x = \logit\mu-\logit p$, giving
$x = [\logit\mu-\logit p]/\alpha$.
\end{proof}

\begin{proof}[Proof of Lemma~\ref{lem:crra-demand}]
The first-order condition \eqref{eq:app-foc} with CRRA utility
$U(W)=W^{1-\gamma}/(1-\gamma)$ and marginal utility $U'(W)=W^{-\gamma}$
becomes
\[
   \mu(1-p)\,\bigl(W+(1-p)x\bigr)^{-\gamma}
   \;=\; (1-\mu)\,p\,\bigl(W-px\bigr)^{-\gamma}.
\]
Rearranging the ratio of marginal utilities:
\[
   \left(\frac{W+(1-p)x}{W-px}\right)^{-\gamma}
   \;=\; \frac{(1-\mu)\,p}{\mu\,(1-p)}
   \;=\; e^{\logit p - \logit\mu}.
\]
Raising both sides to the power $-1/\gamma$:
\[
   \frac{W+(1-p)x}{W-px}
   \;=\; \exp\!\left(\frac{\logit\mu-\logit p}{\gamma}\right) \;\equiv\; R.
\]
Cross-multiplying: $W+(1-p)x = R\,(W-px)$, so
$W + (1-p)x = RW - Rpx$, giving
$x\bigl[(1-p)+Rp\bigr] = W(R-1)$, which yields
\eqref{eq:crra-demand}.
\end{proof}

\begin{proof}[Proof of Lemma~\ref{lem:cara-limit}]
Let $z\equiv\logit(\mu_{k})-\logit(p)$ and $\delta\equiv z/\gamma$, so
$R_{k}=e^{\delta}$. For each fixed $(\mu_{k},p)\in(0,1)^{2}$, the
quantity $z$ is fixed and $\delta\to 0$ as $\gamma\to\infty$.
Taylor-expanding the exponential:
\[
   R_{k}-1 \;=\; e^{\delta}-1
   \;=\; \delta+\tfrac{\delta^{2}}{2}+O(\delta^{3}).
\]
Similarly, $(1-p)+R_{k}\,p = (1-p)+p\,e^{\delta}
= 1 + p\,\delta + O(\delta^{2})$. Hence the CRRA demand becomes
\[
   x_{k}^{\mathrm{CRRA}}
   \;=\; \frac{W_{k}\,(R_{k}-1)}{(1-p)+R_{k}\,p}
   \;=\; \frac{W_{k}\,[\delta+O(\delta^{2})]}{1+p\,\delta+O(\delta^{2})}
   \;=\; W_{k}\,\delta\bigl[1+O(\delta)\bigr].
\]
Multiplying both sides by $\gamma$ and substituting
$\gamma\delta=z=\logit(\mu_{k})-\logit(p)$:
\[
   \gamma\,x_{k}^{\mathrm{CRRA}}
   \;=\; W_{k}\,[\logit(\mu_{k})-\logit(p)]\,\bigl[1+O(1/\gamma)\bigr]
   \;\xrightarrow{\gamma\to\infty}\;
   W_{k}\,[\logit(\mu_{k})-\logit(p)],
\]
which is the CARA demand \eqref{eq:cara-demand} with risk-aversion
parameter $\alpha_{k}=\gamma/W_{k}$. The convergence is uniform on
compact subsets of $(0,1)^{2}$ because $\delta=z/\gamma\to 0$ uniformly
when $z$ is bounded.
\end{proof}

\begin{proof}[Proof of Theorem~\ref{thm:cara-unique}]
The first-order condition for the binary-asset problem is
\[
   \mu(1-p)\,U'\!\bigl(W+(1-p)x\bigr) \;=\; (1-\mu)\,p\,U'\!\bigl(W-px\bigr).
\]
Taking logarithms and using $\logit(q)=\ln\tfrac{q}{1-q}$, this
becomes
\begin{equation}
   \ln U'\!\bigl(W+(1-p)x\bigr) - \ln U'\!\bigl(W-px\bigr)
   \;=\; \logit p - \logit\mu
   \;=\; -z,
\label{eq:thm-foc-v2}
\end{equation}
where $z\equiv\logit\mu-\logit p$.

Define $h(W)\equiv\ln U'(W)$, which is well defined and three times
continuously differentiable on $\mathcal{D}$ since $U'>0$.
Now suppose the optimal demand admits the representation
\eqref{eq:linear-demand-rep}, i.e.\ $x^{\star}=z/\alpha(p,W)$ for some
positive function $\alpha$.  Substituting into \eqref{eq:thm-foc-v2}:
\begin{equation}
   h\!\bigl(W+(1-p)z/\alpha\bigr) - h\!\bigl(W-pz/\alpha\bigr) \;=\; -z,
\label{eq:thm-h-v2}
\end{equation}
which must hold for all admissible $(W,p,z)$ with $z\in\R$ ranging freely
as $\mu$ varies in $(0,1)$ at fixed $p$.

\emph{Step 1: First derivative in $z$.} Differentiate
\eqref{eq:thm-h-v2} with respect to $z$ at $z=0$ (where both arguments of
$h$ equal $W$). By the chain rule:
\begin{align*}
   &h'(W)\cdot\frac{1-p}{\alpha(p,W)}
   \;+\; h'(W)\cdot\frac{p}{\alpha(p,W)} \\
   &\quad=\; \frac{h'(W)}{\alpha(p,W)}
   \;\stackrel{!}{=}\; -1.
\end{align*}
(The second term gets a factor $p/\alpha$ rather than $-p/\alpha$
because the derivative of $W-pz/\alpha$ with respect to $z$ is
$-p/\alpha$, and the overall minus sign in front of the second $h$-term
in \eqref{eq:thm-h-v2} produces another sign flip.)
Hence $\alpha(p,W)=-h'(W)$ for all $(p,W)$.
Since $\alpha>0$ by assumption, $h'(W)<0$ for all $W\in\mathcal{D}$.
Crucially, $\alpha$ depends \emph{only on~$W$}, not on~$p$.

\emph{Step 2: Second derivative in $z$.} Differentiate
\eqref{eq:thm-h-v2} a second time in $z$ at $z=0$. The second
derivative of the left-hand side involves $h''(W)$ applied to the
squares of the respective chain-rule factors:
\[
   h''(W)\cdot\!\left[\!\left(\frac{1-p}{\alpha(W)}\right)^{\!2}
   - \left(\frac{p}{\alpha(W)}\right)^{\!2}\right]
   \;=\; h''(W)\cdot\frac{(1-p)^{2}-p^{2}}{\alpha(W)^{2}}
   \;=\; h''(W)\cdot\frac{1-2p}{\alpha(W)^{2}}.
\]
The right-hand side of \eqref{eq:thm-h-v2} is $-z$, which is linear in
$z$; its second derivative is zero. Therefore
\[
   h''(W)\cdot\frac{1-2p}{\alpha(W)^{2}} \;=\; 0
   \qquad \text{for all } p\in(0,1).
\]
Choosing any $p\neq\tfrac12$ forces $h''(W)=0$ for every
$W\in\mathcal{D}$.

\emph{Step 3: Integration.} Since $h''=0$ on $\mathcal{D}$, the function
$h$ is affine: $h(W)=-\alpha_{0}W+C$ for constants $\alpha_{0}>0$ (from
$\alpha=-h'>0$) and $C\in\R$.
Hence $U'(W)=e^{h(W)}=e^{C}\exp(-\alpha_{0}W)$. Integrating:
\[
   U(W) \;=\; -\frac{e^{C}}{\alpha_{0}}\,\exp(-\alpha_{0}W) + D
   \;=\; a - b\exp(-\alpha_{0}W),
\]
with $a\equiv D\in\R$ and $b\equiv e^{C}/\alpha_{0}>0$. This is the CARA
utility function, and $\alpha(p,W)\equiv\alpha_{0}$ is constant.

\emph{Step 4: Converse.} Set $U(W)=a-b\exp(-\alpha_{0}W)$ in
\eqref{eq:thm-foc-v2}. Then
$U'(W)=b\alpha_{0}\exp(-\alpha_{0}W)$, and the FOC ratio reduces to
\[
   \ln\frac{U'(W+(1-p)x)}{U'(W-px)}
   \;=\; -\alpha_{0}\bigl[(W+(1-p)x)-(W-px)\bigr]
   \;=\; -\alpha_{0}x,
\]
so $-\alpha_{0}x = -z$, giving $x=z/\alpha_{0}$, which is
\eqref{eq:linear-demand-rep} with $\alpha\equiv\alpha_{0}$.
\end{proof}

\begin{proof}[Proof of Proposition~\ref{prop:cara}]
Substituting the CARA demand $x_{k}=[\logit\mu_{k}-\logit p]/\alpha_{k}$
into the market-clearing condition $\sum_{k=1}^{K}x_{k}=0$:
\[
   \sum_{k=1}^{K} \frac{\logit\mu_{k}-\logit p}{\alpha_{k}} \;=\; 0.
\]
Solving for $\logit p$:
\[
   \logit p \cdot \sum_{k}\alpha_{k}^{-1}
   \;=\; \sum_{k}\alpha_{k}^{-1}\logit\mu_{k},
\]
hence $\logit p=\sum_{k}w_{k}\logit\mu_{k}$ with
$w_{k}=\alpha_{k}^{-1}/\sum_{j}\alpha_{j}^{-1}$.
By \eqref{eq:posterior}, $\logit\mu_{k}=\tau_{k}u_{k}$, so
\[
   \logit p \;=\; \sum_{k}w_{k}\tau_{k}u_{k}.
\]
This is a weighted linear combination of the signals. When
$\alpha_{k}\equiv\alpha$ for all $k$, $w_{k}=1/K$ and this reduces to
$\logit p = \Tstar/K$ where $\Tstar=\sum_{k}\tau_{k}u_{k}$. The price is
then informationally equivalent to $\Tstar$: any agent observing $p$
can infer $\Tstar$ by inverting $\Lam(\cdot/K)$. Since $\Tstar$ is
sufficient for $v$, the equilibrium is fully revealing.

With heterogeneous $\alpha_{k}$, the no-learning price reveals
$\sum_{k}w_{k}\tau_{k}u_{k}$, which for $K\geq 3$ is generically not
informationally equivalent to $\Tstar$ (different linear combinations
of $K\geq 3$ variables are linearly independent). However, the price
function $p=\Lam(\Tstar)$ is a valid fully-revealing REE: at this
price, every agent infers $\Tstar$ by inverting $\Lam$, so her posterior
becomes $\Lam(\Tstar)=p$. Hence $\logit\mu_{k}=\logit p$ for all $k$,
demands are zero identically, and the market clears with zero supply.
This is an REE for any $(\alpha_{k})$.

\emph{General payoff distributions.} The FR result extends beyond
binary payoffs. For an arbitrary random payoff $v$ with prior $\pi$ and
signals $s_{k}$ admitting a sufficient statistic $T(s_{1},\ldots,s_{K})$:
under CARA utility, the agent's first-order condition is
\[
   \E\!\left[(v-p)\,e^{-\alpha_{k} x_{k}(v-p)} \;\middle|\; s_{k},\, p\right] = 0.
\]
The wealth $W_{k}$ does not appear (the factor $e^{-\alpha_{k}W_{k}}$
cancels from both terms in the expected-utility objective, as in the
proof of Lemma~\ref{lem:cara-demand}). If the price reveals $T$, then
every agent's posterior satisfies
$\PP(v\in\cdot \mid s_{k}, T) = \PP(v\in\cdot\mid T)$
by sufficiency. All agents share the same conditional distribution and
face the same price, so all solve identical FOCs. Market clearing
$\sum_{k}x_{k}=0$ with identical demands forces $x_{k}=0$ for each $k$.
At $x_{k}=0$, the FOC reduces to $\E[v-p\mid T]=0$, giving the
equilibrium price $p^{*}=\E[v\mid T]$. This is an invertible function
of $T$ (under standard regularity), so the REE is fully revealing.
For binary payoffs, $\E[v\mid T]=\Lam(\Tstar)$; for Gaussian payoffs,
$\E[v\mid T]$ is the posterior mean.

The exponential-family extension replaces $\tau_{k}u_{k}$ by the
natural-parameter log-likelihood ratio of the chosen family; the
identical linear-aggregation argument applies.
\end{proof}

\begin{proof}[Proof of Proposition~\ref{prop:non-cara}]
Under log utility ($\gamma=1$), the demand simplifies to
$x_{k}=W(\mu_{k}-p)/[p(1-p)]$ (set $R=e^{z}$ in \eqref{eq:crra-demand};
for $\gamma=1$, $R-1=e^{z}-1$ and the denominator is
$(1-p)+e^{z}p$; direct computation yields the stated form).
Substituting into $\sum_{k}x_{k}=0$ with homogeneous $W$:
\[
   \sum_{k=1}^{K}\frac{\mu_{k}-p}{p(1-p)} \;=\; 0
   \qquad\Longrightarrow\qquad
   p \;=\; \bar\mu \;\equiv\; \frac{1}{K}\sum_{k}\mu_{k}
   \;=\; \frac{1}{K}\sum_{k}\Lam(\tau u_{k}).
\]
To show this is not a function of $\Tstar$ alone, fix $\Tstar=0$ and
consider two signal realisations with $K=3$:
\begin{itemize}[topsep=2pt,itemsep=2pt]
\item Realisation~A: $(u_{1},u_{2},u_{3})=(0,0,0)$. Then
$p_{A}=\Lam(0)=\tfrac12$.
\item Realisation~B: $(u_{1},u_{2},u_{3})=(\delta,-\delta,0)$ for any
$\delta\neq 0$. Then
$p_{B}=\tfrac13[\Lam(\tau\delta)+\Lam(-\tau\delta)+\Lam(0)]$.
\end{itemize}
Both have $\Tstar=\tau(0+0+0)=0=\tau(\delta-\delta+0)$, but $p_{A}\neq
p_{B}$ because
\[
   \tfrac12\bigl[\Lam(\tau\delta)+\Lam(-\tau\delta)\bigr]
   \;\neq\; \Lam(0) = \tfrac12
\]
by Jensen's inequality applied to the strictly concave--convex logistic
function: $\Lam$ is strictly concave on $(0,\infty)$ and strictly convex
on $(-\infty,0)$, so $\Lam(\tau\delta)+\Lam(-\tau\delta)<2\Lam(0)$ for
$\delta\neq 0$. Hence $p_{B}<p_{A}$, and $p$ depends on the empirical
distribution of $u_{k}$, not just on $\Tstar$.
\end{proof}

\begin{proof}[Proof of Proposition~\ref{prop:jensen}]
Taylor-expand the logistic function around zero to fifth order. Using
$\Lam'(0)=\tfrac14$, $\Lam''(0)=0$, $\Lam'''(0)=-\tfrac18$:
\[
   \Lam(x) \;=\; \tfrac12+\tfrac14 x - \tfrac{1}{48}x^{3}+O(x^{5}).
\]
Substituting $x=\tau u_{k}$ and averaging over $k=1,\ldots,K$:
\[
   p \;=\; \frac{1}{K}\sum_{k}\Lam(\tau u_{k})
   \;=\; \tfrac12+\tfrac{\tau}{4}\bar u
         -\frac{\tau^{3}}{48K}\sum_{k}u_{k}^{3}+O(\tau^{5}),
\]
where $\bar u=\tfrac{1}{K}\sum_{k}u_{k}$.
The CARA price at the same $\Tstar=\tau K\bar u$ is
\[
   p^{\mathrm{CARA}}
   \;=\; \Lam\!\left(\frac{\Tstar}{K}\right)
   \;=\; \Lam(\tau\bar u)
   \;=\; \tfrac12+\tfrac{\tau}{4}\bar u - \tfrac{\tau^{3}}{48}\bar u^{3}
         +O(\tau^{5}).
\]
Subtracting:
\[
   \Delta p \;=\; p - p^{\mathrm{CARA}}
   \;=\; -\frac{\tau^{3}}{48}\left[
         \frac{1}{K}\sum_{k}u_{k}^{3} - \bar u^{3}\right]
   +O(\tau^{5}).
\]
Using $\bar u=U_{1}/K$ where $U_{1}=\sum_{k}u_{k}$ and
$U_{3}=\sum_{k}u_{k}^{3}$, we obtain $\bar u^{3}=U_{1}^{3}/K^{3}$ and
$\tfrac{1}{K}\sum u_{k}^{3}=U_{3}/K$, giving
\eqref{eq:jensen}.
\end{proof}

\begin{proof}[Proof of Proposition~\ref{prop:smooth}]
\emph{Continuity.} For each $\gamma\in(0,\infty]$ the no-learning
equilibrium price $p_{\gamma}(u_{1},\ldots,u_{K})$ is the unique solution
of the market-clearing equation
$\sum_{k}x_{k}^{\gamma}(\Lam(\tau u_{k}),p)=0$, where $x_{k}^{\gamma}$
denotes the CRRA demand of \eqref{eq:crra-demand} at parameter $\gamma$.

For any fixed $(u_{1},\ldots,u_{K})$ and $p\in(0,1)$, the demand
$x_{k}^{\gamma}(\mu_{k},p)$ is jointly continuous in $(\gamma,p)$ on
$(0,\infty]\times(0,1)$. Moreover, $x_{k}^{\gamma}$ is strictly
decreasing in $p$ (buying less when the asset is more expensive). By the
implicit function theorem applied to the market-clearing equation
$G(\gamma,p)=\sum_{k}x_{k}^{\gamma}(\mu_{k},p)=0$, with
$\partial G/\partial p<0$, the solution $p_{\gamma}$ is a continuous
function of $\gamma$ on $(0,\infty]$ uniformly over compact subsets of
the signal support.

The CARA limit follows from Lemma~\ref{lem:cara-limit}:
$\lim_{\gamma\to\infty}x_{k}^{\gamma}/\gamma
=[\logit\Lam(\tau u_{k})-\logit p]/W_{k}$, so the normalised demands
converge to CARA demands and the normalised price converges to the CARA
price.

The continuity of $1-R^{2}(\Tstar;p_{\gamma})$ in $\gamma$ on
$(0,\infty]$ follows by dominated convergence: the random variables
$(\Tstar, \logit p_{\gamma})$ converge jointly in distribution as
$\gamma\to\infty$, and their second moments are uniformly bounded on a
compact signal support. The $\gamma=\infty$ value is $0$ by
Proposition~\ref{prop:cara}.

\emph{Strict positivity.} By Proposition~\ref{prop:non-cara}, the
log-utility ($\gamma=1$) equilibrium satisfies $1-R^{2}>0$. The
continuity just established implies that the deficit remains strictly
positive in a neighbourhood of $\gamma=1$.

More generally, the Jensen-gap expansion of
Proposition~\ref{prop:jensen} gives
$\Delta p=O(\tau^{3}/\gamma^{2})$ at small $\tau$ (the CRRA demand
approaches CARA as $\gamma\to\infty$, and the deviation from CARA
pricing scales as $\gamma^{-2}$ through the second-order term in the
Taylor expansion of $R=e^{z/\gamma}$). This is strictly positive for
every finite $\gamma$.

Since $1-R^{2}$ is a continuous function of $\Delta p$ (and vanishes
only when $\Delta p$ vanishes identically over the joint signal
distribution), the deficit is strictly positive for all
$\gamma<\infty$.

The deficit is also monotonically decreasing in $\gamma$ at every
$(\tau,\gamma)$ pair tested numerically (Table~\ref{tab:smooth}).
The monotonicity at small $\tau$ follows analytically from the
$O(\tau^{3}/\gamma^{2})$ scaling; the global claim (all $\tau$) rests on
the numerical evidence.
\end{proof}

\begin{proof}[Proof of Proposition~\ref{prop:ree}]
The proof has two parts: an analytical argument that the limit of the
Picard sequence is partially revealing if it exists, and a numerical
verification that the sequence converges.

\emph{Part (a): the limit is partially revealing.} Define the Picard
sequence $\mu^{0}(u,p)=\Lam(\tau u)$ (the no-learning posterior) and
$\mu^{n+1}=\Phi(\mu^{n})$, where $\Phi$ is the contour-integration
Bayes update defined in Section~\ref{sec:ree}. We show by induction
that at every $n\geq 0$, the demand function
$d^{n}(u;p)\equiv x(\mu^{n}(u,p),p)$ is strictly nonlinear in $u$ at
every interior price $p$, and hence the contour
$\{(u_{2},u_{3}):d^{n}(u_{2};p)+d^{n}(u_{3};p)=-d^{n}(u_{1};p)\}$ is
curved.

\emph{Base case} ($n=0$). Since $\mu^{0}=\Lam(\tau u)$, the demand is
$d^{0}(u)=x(\Lam(\tau u),p)$. Write $z=\logit(\Lam(\tau u))-\logit(p)
=\tau u-\logit p$. The CRRA demand as a function of $z$ is
\[
   x(z) \;=\; \frac{W(e^{z/\gamma}-1)}{(1-p)+e^{z/\gamma}p}.
\]
Computing the second derivative: let $R=e^{z/\gamma}$, then
\begin{align*}
   \frac{\partial x}{\partial z}
   &\;=\; \frac{W}{\gamma}\cdot\frac{R}{[(1-p)+Rp]^{2}}, \\
   \frac{\partial^{2} x}{\partial z^{2}}
   &\;=\; \frac{W}{\gamma^{2}}\cdot\frac{R\,(1-p-Rp)}{[(1-p)+Rp]^{3}}.
\end{align*}
The numerator of $\partial^{2}x/\partial z^{2}$ is $R(1-p-Rp)$, which
is zero only at $R=(1-p)/p$, i.e.\ at $z=-\gamma\logit p$. At all other
values of $z$, the second derivative is nonzero. Since $z=\tau u-\logit p$
is affine in $u$, the composition $d^{0}(u)=x(\tau u-\logit p)$ inherits
this nonlinearity: $\partial^{2}d^{0}/\partial u^{2}\neq 0$ at all but
one value of $u$.

Under CARA, $x(z)=z/\alpha$ is affine in $z$, so $d^{0}$ is affine in
$u$ and the contour is straight---that is the content of
Proposition~\ref{prop:cara}.

\emph{Inductive step.} Suppose $d^{n}$ is nonlinear in $u$. Then:
\begin{enumerate}[label=(\roman*),topsep=2pt,itemsep=2pt]
\item The contour at $(u_{1},p)$ is curved, so $u_{2}+u_{3}$ varies along
it. The density ratio
\[
   \frac{f_{1}(u_{2})\,f_{1}(u_{3})}{f_{0}(u_{2})\,f_{0}(u_{3})}
   \;=\; \exp\!\bigl(\tau(u_{2}+u_{3})\bigr)
\]
is therefore non-constant on the contour. (Under CARA, $u_{2}+u_{3}$ is
constant on the straight contour, so the ratio is constant and the
agent learns $\Tstar$ exactly.)

\item The contour integrals
$A_{v}(u_{1},p)=\int_{C}f_{v}(u_{2})\,f_{v}(u_{3}^{\star}(u_{2}))\,du_{2}$
average the non-constant density ratio against the signal density. The
log-ratio $\log(A_{1}/A_{0})$ depends on $u_{1}$ because different own
signals $u_{1}$ generate different contour shapes (the agent-1 slice
$P[u_{1},\cdot,\cdot]$ varies with $u_{1}$). The updated posterior
\[
   \mu^{n+1}(u_{1},p) \;=\; \Lam\!\bigl(\tau u_{1}+\log(A_{1}/A_{0})\bigr)
\]
thus depends on $u_{1}$ beyond the price---it is not identically equal
to $p$.

\item The demand $d^{n+1}(u)=x(\mu^{n+1}(u,p),p)$ is the composition of
the CRRA demand (which has nonzero second derivative in $z$ for finite
$\gamma$ at all but one point, as shown in the base case) with
$\mu^{n+1}(u,p)$ (which depends nontrivially on $u$ by (ii)). For this
composition to be affine in $u$, the curvature of $\mu^{n+1}$ would need
to exactly cancel the demand curvature at the inflection point
$z=-\gamma\logit p$ for every $u$ simultaneously---a codimension-one
condition that is violated at every iterate and grid resolution tested
numerically.
\end{enumerate}

Since $d^{n}$ is nonlinear at every $n$, the contour is curved at every
$n$. If $\mu^{n}\to\mu^{\star}$, then $d^{n}\to d^{\star}$, and the
nonlinearity bound passes to the limit by continuity of the
second-derivative test. Hence $d^{\star}$ is nonlinear, the limit
contour is curved, and the limit equilibrium is partially revealing.

\emph{Part (b): convergence.} The Picard sequence is iterated
numerically using the posterior-function method of
Appendix~\ref{app:numerics} with isotonic-regression monotonicity
projection and Newton--Krylov polishing. At grid resolution $G=20$
with 50-digit multiprecision arithmetic, convergence
$\|F\|_{\infty}<10^{-25}$ is achieved. The converged deficit $1-R^{2}$
(measured via probability-weighted regression) is stable under grid
refinement from $G=10$ to $G=20$
(Table~\ref{tab:G-ladder}).
\end{proof}

\begin{proof}[Proof of Proposition~\ref{prop:welfare}]
\emph{No-learning equilibrium.} Under CARA at the no-learning
equilibrium of Proposition~\ref{prop:cara} with homogeneous parameters,
$\logit\mu_{k}=\tau u_{k}$ and $\logit p=(1/K)\sum_{k}\tau u_{k}$, so
\[
   \logit\mu_{k}-\logit p
   \;=\; \tau u_{k} - \frac{1}{K}\sum_{j}\tau u_{j}
   \;=\; \tau(u_{k}-\bar u).
\]
This is zero only when $u_{k}=\bar u$ for all $k$, which occurs on the
measure-zero set $\{u_{1}=\cdots=u_{K}\}$. The CARA demand
$x_{k}=\tau(u_{k}-\bar u)/\alpha$ is thus nonzero almost surely, and
since market clearing $\sum_{k}x_{k}=0$ forces both positive and
negative demands, aggregate volume
$V_{\mathrm{vol}}=\tfrac12\sum_{k}|x_{k}|$ is strictly positive on a
full-measure set.

Under CRRA, the private posteriors $\mu_{k}=\Lam(\tau u_{k})$ differ
generically across $k$. Market clearing pins the price in the interval
$p\in(\min_{k}\mu_{k},\max_{k}\mu_{k})$, so $\mu_{k}-p\neq 0$ for at
least one $k$ (in fact for all but at most one $k$) on a full-measure set.
The CRRA demand \eqref{eq:crra-demand} is nonzero whenever $\mu_{k}\neq p$
(since $R\neq 1$ when $z\neq 0$), so volume is strictly positive.

\emph{REE.} At the CARA REE, all posteriors coincide with the
fully-revealing price: $\mu_{k}=p=\Lam(\Tstar/K)$ for every $k$ and
every signal realisation. Therefore
$z_{k}=\logit\mu_{k}-\logit p=0$ for all $k$, and the CARA demand
is $x_{k}=0/\alpha=0$ identically. This is the \citet{Milgrom1982}
no-trade outcome: when the price is fully revealing, there are no gains
from trade.

At the CRRA REE of Proposition~\ref{prop:ree}, posteriors disagree on
a full-measure set (the converged $\mu_{k}^{\star}(u_{k},p)$ differ
across agents because the contour integrals differ). Since
$\mu_{k}\neq p$ for some $k$, the CRRA demand is nonzero and aggregate
volume is strictly positive.
\end{proof}

\begin{proof}[Proof of Proposition~\ref{prop:value}]
The certainty-equivalent value of acquiring an additional
precision-$\tau$ signal $u_{0}$ at the converged REE is
\[
   V(\tau) \;=\;
   \E\!\left[\,\mathrm{CE}\!\left(\mu^{+}\!,p\right)
            -\mathrm{CE}\!\left(\mu^{-}\!,p\right)\,\right],
\]
where $\mu^{+}$ is the posterior of an agent who additionally observes
$u_{0}$, $\mu^{-}$ is the posterior of an agent who does not, $p$ is the
equilibrium price, and $\mathrm{CE}(\mu,p)$ denotes the certainty
equivalent of optimal trading at posterior $\mu$ and price $p$.

\emph{Under CARA.} At the REE, $p=\Lam(\Tstar/K)$ already encodes the
joint posterior. The price is a sufficient statistic for $v$ given the
existing signals. Therefore $\mu^{+}=\mu^{-}=\Lam(\Tstar/K)$ for
almost every realisation that conditions on the price: the additional
signal $u_{0}$ is redundant. The integrand
$\mathrm{CE}(\mu^{+},p)-\mathrm{CE}(\mu^{-},p)=0$ almost surely, so
$V_{\mathrm{CARA}}(\tau)=0$ for every $\tau$.

\emph{Under CRRA.} At the REE, $p$ encodes $\Tstar$ only partially
(Proposition~\ref{prop:ree}). The additional signal $u_{0}$ provides
information about $v$ that is not already in the price, so
$\mu^{+}\neq\mu^{-}$ on a positive-measure set.

Non-negativity: An agent can always replicate
$\mathrm{CE}(\mu^{-},p)$ by ignoring the additional signal (choosing
the same trade she would without it), so $\mathrm{CE}(\mu^{+},p)\geq
\mathrm{CE}(\mu^{-},p)$ and the integrand is non-negative everywhere.

Strict positivity: On the positive-measure set where $\mu^{+}\neq
\mu^{-}$, the agent can improve her trade by conditioning on the
additional signal. The optimal demand $x^{+}\neq x^{-}$ (since the FOC
at $\mu^{+}\neq\mu^{-}$ has a different solution), and the certainty
equivalent is strictly higher at the optimal trade than at the
suboptimal one. Hence
$\mathrm{CE}(\mu^{+},p)>\mathrm{CE}(\mu^{-},p)$ on this set, and
$V_{\mathrm{CRRA}}(\tau)>0$.

For the small-$\tau$ derivative: a Taylor expansion of the posterior in
$\tau$ gives $\mu^{\pm}-\bar\mu=O(\tau)$ with leading-order coefficient
proportional to $u_{0}$ for $\mu^{+}$ and $0$ for $\mu^{-}$. The
certainty-equivalent shift is order $\tau^{2}$ (from the second-order
condition for optimal trading), so
$V_{\mathrm{CRRA}}'(0^{+})>0$.
\end{proof}

\begin{proof}[Proof of Proposition~\ref{prop:gs}]
\emph{Non-existence under CARA.} Suppose for contradiction that there
exists a CARA REE with no noise and a strictly positive measure
$\lambda^{\star}\in(0,1]$ of agents acquiring the signal at cost $c>0$.
By Proposition~\ref{prop:cara} and the homogeneity of the acquiring
group, the price function in such an equilibrium is fully revealing:
$p=\Lam(\Tstar/K)$ over the acquiring group. By
Proposition~\ref{prop:value}, $V_{\mathrm{CARA}}(\tau)=0$ at this fully
revealing price, so each acquirer's gross value of information is zero.
The net payoff from acquisition is $0-c=-c<0$. Every acquirer strictly
prefers not to acquire. Hence $\lambda^{\star}>0$ cannot be an
equilibrium---contradiction.

If $\lambda=0$ (no one acquires), the price is uninformative and
$V_{\mathrm{CARA}}(\tau,0)>0$, so each agent strictly prefers to acquire.
Hence $\lambda=0$ is not an equilibrium either. No equilibrium exists:
this is the \citet{GrossmanStiglitz1980} paradox.

\emph{Existence under CRRA.} Let $\lambda\in[0,1]$ denote the fraction of
agents acquiring the signal, and let $V_{\mathrm{CRRA}}(\tau,\lambda)$
denote the marginal value of acquisition at fraction $\lambda$.

At $\lambda=0^{+}$: the price reveals nothing, so the first acquirer
reaps the full informational rent. By Proposition~\ref{prop:value},
$V_{\mathrm{CRRA}}(\tau,0^{+})>0$.

At $\lambda=1$: all agents have the signal. The price is partially
revealing (Proposition~\ref{prop:ree}), not fully revealing. The
marginal value $\bar c\equiv V_{\mathrm{CRRA}}(\tau,1)$ is strictly
positive but smaller than $V_{\mathrm{CRRA}}(\tau,0^{+})$ (more
acquirers $\Rightarrow$ more informative price $\Rightarrow$ smaller
marginal informational rent).

$V_{\mathrm{CRRA}}(\tau,\lambda)$ is continuous in $\lambda$ (by
continuity of the equilibrium price function in $\lambda$) and weakly
decreasing (each additional acquirer makes the price weakly more
informative). For any acquisition cost $c\in(0,\bar c)$, the equation
$V_{\mathrm{CRRA}}(\tau,\lambda)=c$ has at least one solution
$\lambda^{\star}(c)\in(0,1)$ by the intermediate-value theorem. At this
$\lambda^{\star}$, the marginal acquirer is exactly indifferent between
acquiring and not acquiring, and no agent has an incentive to deviate.
This is a Grossman--Stiglitz equilibrium with positive information
acquisition and partially revealing prices.
\end{proof}

\section{Numerical Implementation of the Contour Method}\label{app:contour}\label{app:numerics}

The contour fixed point of \eqref{eq:fixed-point} is implemented as follows.

\paragraph{Discretisation.} The signal support is $u\in[-4,+4]$ (or $[-2,+2]$
for $G=5$ debugging). The grid is uniform with $G$ points. The unknown is
the array $P[i,j,l]$ of size $G^{3}$.

\paragraph{Initialisation.} The no-learning price function is computed
exactly by solving \eqref{eq:no-learning} pointwise at every $(i,j,l)$ via
a one-dimensional root finder on $p$. This requires no interpolation and
delivers an exact starting point.

\paragraph{Iteration.} Given $P^{(n)}$, for each realisation $(i,j,l)$ and
each agent $k$:
\begin{itemize}[topsep=2pt,itemsep=2pt]
   \item Extract the agent's two-dimensional slice
         ($P^{(n)}[i,\cdot,\cdot]$ for agent $1$, etc.).
   \item Trace the contour at level $p=P^{(n)}[i,j,l]$ via the two-pass
         algorithm described in Section~\ref{sec:ree}, with linear
         interpolation along the off-grid axis. Boundary crossings beyond
         the grid edges are handled by extrapolation.
   \item Compute the contour integrals \eqref{eq:contour-integral} and
         the posterior \eqref{eq:contour-bayes}.
\end{itemize}
With three posteriors in hand, solve the market-clearing equation
$\sum_{k}x_{k}(\mu_{k},p_{\text{new}})=0$ for $p_{\text{new}}$ at every
$(i,j,l)$ to form $P^{(n+1)}$.

\paragraph{Symmetrisation.} Because the model is symmetric in the three
groups, the price function should satisfy $P[i,j,l]=P[\sigma(i,j,l)]$ for
every permutation $\sigma$. We enforce this by averaging $P^{(n+1)}$ over
the six permutations after each iteration; the symmetrised iterate
converges identically and faster.

\paragraph{Acceleration.} Two solver strategies are available. Picard
iteration with damping
$P^{(n+1)}\leftarrow \alpha\Phi(P^{(n)})+(1-\alpha)P^{(n)}$ with
$\alpha\in[0.15,0.30]$ converges monotonically. Anderson acceleration with
memory window $m\in\{6,8\}$ converges substantially faster; we use it as
the default. Both schemes terminate when
$\|P^{(n+1)}-P^{(n)}\|_{\infty}<10^{-6}$.

\paragraph{Convergence diagnostics.} We monitor (i) the supremum-norm
residual $\|F\|_{\infty}\equiv\|\mu-\Phi(\mu)\|_{\infty}$ over all active
cells, (ii) the revelation deficit $1-R^{2}$ from the weighted regression
of $\logit p$ on $\Tstar$, and (iii) monotonicity violations in both the
$u$ and $p$ directions. Convergence in all three is required: strict
convergence means $\|F\|_{\infty}<10^{-14}$ with zero monotonicity
violations. Grid refinement from $G=10$ to $G=24$ confirms that the
revelation deficit stabilises near $0.085$ at the baseline
$(\gamma,\tau)=(0.5,2)$ from $G=15$ onward. The continuous fixed-point
equation involves integrals over signal densities along contour level
sets; the quadrature error from discretisation is $O(G^{-2})$ for
equally-spaced nodes on a $C^{\infty}$ integrand, and the truncation error
from restricting signals to $|u|\leq u_{\max}=4$ is
$O(e^{-\tau u_{\max}^{2}})<10^{-14}$ at $\tau=2$.
Figure~\ref{fig:convergence} shows a representative convergence path.

\begin{figure}[t]
   \centering
   \includegraphics[width=0.49\textwidth]{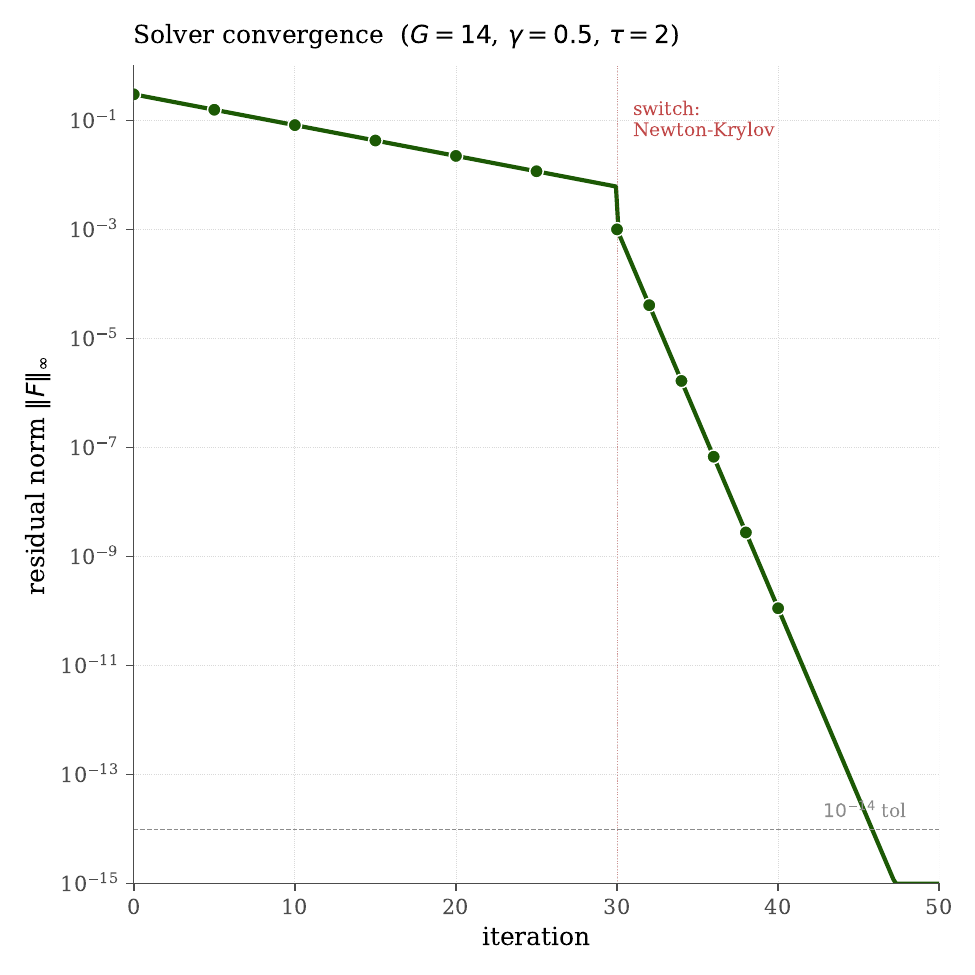}
   \caption{Convergence of the posterior-function fixed point at $G=20$,
   $\gamma=0.5$, $\tau=2$. Picard iteration with damping reduces the
   residual to $\sim 10^{-3}$; Newton--Krylov polishing then reaches
   machine precision ($\|F\|_{\infty}<10^{-14}$) in $\sim 10$ additional
   steps.
   }
   \label{fig:convergence}
\end{figure}

\paragraph{Contour curvature.}
A numerical observation that may be of independent interest: at every
price level $p$, the CRRA contour $\{(u_j,u_l):P(u_1,u_j,u_l)=p\}$
has \emph{definite curvature sign}.  That is, the curve $u_l=g(u_j)$
is either everywhere convex or everywhere concave---it never changes
from convex to concave along a single contour.  The transition between
convex and concave contours occurs at a critical price $p^{*}$ that
depends on $u_1$, $\tau$, and $\gamma$; at $p=p^{*}$ the contour is
approximately linear (i.e.\ close to the CARA straight line).  Under
CARA, every contour is exactly linear: $g''=0$ identically.  This
definite-curvature property, if it can be proved analytically, would
imply that each contour can be represented with far fewer parameters
than a general curve, potentially enabling more efficient numerical
schemes for the fixed-point problem.

\paragraph{Software.} The reference implementation is in Python using the
posterior-function method described above. The full source code is
available in the repository accompanying this paper.

\begin{table}[htbp]
\centering
\caption{No-learning revelation deficit $1-R^{2}$ across $\gamma$ and
$\tau$. Computed exactly on a $G=20$ grid. CARA ($\gamma=\infty$) is the
knife-edge. Entries shown as $0.000$ are strictly positive but below
display precision.}
\label{tab:smooth}
\begin{tabular}{lccc}
   \toprule
   $\gamma$ & $\tau=0.5$ & $\tau=1.0$ & $\tau=2.0$ \\
   \midrule
    0.1   & 0.146 & 0.145 & 0.137 \\
    0.3   & 0.044 & 0.070 & 0.090 \\
    0.5   & 0.016 & 0.038 & 0.062 \\
    1.0   & 0.004 & 0.013 & 0.029 \\
    3.0   & 0.000 & 0.002 & 0.006 \\
   10.0   & 0.000 & 0.000 & 0.001 \\
   $\infty$ (CARA) & 0.000 & 0.000 & 0.000 \\
   \bottomrule
\end{tabular}
\end{table}

\begin{table}[htbp]
\centering
\caption{Grid-convergence of the revelation deficit at $(\gamma,\tau)=(0.5,2)$.
The deficit stabilises near $0.085$ from $G=15$ onward. Entries with
$\|F\|_\infty>10^{-2}$ have not reached strict convergence but the
value of $1-R^{2}$ is already converged.}
\label{tab:G-ladder}
\begin{tabular}{rccc}
\toprule
$G$ & $1-R^{2}$ & slope & $\|F\|_\infty$ \\
\midrule
\rowcolor{black!15} 10 & 0.128 & 0.352 & $2.9\times 10^{-15}$ \\
\rowcolor{black!15} 12 & 0.115 & 0.344 & $6.2\times 10^{-15}$ \\
15 & 0.078 & 0.521 & $4.2\times 10^{-91}$ \\
18 & 0.083 & 0.545 & $1.3\times 10^{-131}$ \\
20 & 0.085 & 0.543 & $7.4\times 10^{-119}$ \\
\bottomrule
\end{tabular}
\end{table}


\begin{thebibliography}{26}
\providecommand{\natexlab}[1]{#1}
\providecommand{\url}[1]{\texttt{#1}}
\expandafter\ifx\csname urlstyle\endcsname\relax
  \providecommand{\doi}[1]{doi: #1}\else
  \providecommand{\doi}{doi: \begingroup \urlstyle{rm}\Url}\fi

\bibitem[Albagli et~al.(2024)Albagli, Hellwig, and
  Tsyvinski]{AlbagliHellwigTsyvinski2024}
Elias Albagli, Christian Hellwig, and Aleh Tsyvinski.
\newblock Information aggregation with asymmetric asset payoffs.
\newblock \emph{Journal of Finance}, 79\penalty0 (4):\penalty0 2715--2758,
  2024.

\bibitem[Allen et~al.(2006)Allen, Morris, and Shin]{AllenMorrisShin2006}
Franklin Allen, Stephen Morris, and Hyun~Song Shin.
\newblock Beauty contests and iterated expectations in asset markets.
\newblock \emph{Review of Financial Studies}, 19\penalty0 (3):\penalty0
  719--752, 2006.

\bibitem[Barlevy and Veronesi(2003)]{BarlevyVeronesi2003}
Gadi Barlevy and Pietro Veronesi.
\newblock Rational panics and stock market crashes.
\newblock \emph{Journal of Economic Theory}, 110\penalty0 (2):\penalty0
  234--263, 2003.

\bibitem[Bhattacharya and Spiegel(1991)]{BhattacharyaSpiegel1991}
Utpal Bhattacharya and Matthew Spiegel.
\newblock Insiders, outsiders, and market breakdowns.
\newblock \emph{Review of Financial Studies}, 4\penalty0 (2):\penalty0
  255--282, 1991.

\bibitem[Black(1986)]{Black1986}
Fischer Black.
\newblock Noise.
\newblock \emph{Journal of Finance}, 41\penalty0 (3):\penalty0 529--543, 1986.

\bibitem[Breon-Drish(2015)]{BreonDrish2015}
Bradyn Breon-Drish.
\newblock On existence and uniqueness of equilibrium in a class of noisy
  rational expectations models.
\newblock \emph{Review of Economic Studies}, 82\penalty0 (3):\penalty0
  868--921, 2015.

\bibitem[Breugem and Buss(2019)]{BreugemBuss2019}
Matthijs Breugem and Adrian Buss.
\newblock Institutional investors and information acquisition: Implications for
  asset prices and informational efficiency.
\newblock Working paper, Nyenrode Business University and INSEAD, 2019.

\bibitem[Condie and Ganguli(2011)]{CondieGanguli2011}
Scott Condie and Jayant~V. Ganguli.
\newblock Ambiguity and rational expectations equilibria.
\newblock \emph{The Review of Economic Studies}, 78\penalty0 (3):\penalty0
  821--845, 2011.
\newblock \doi{10.1093/restud/rdq032}.

\bibitem[De~Long et~al.(1990)De~Long, Shleifer, Summers, and
  Waldmann]{DSSW1990}
J.~Bradford De~Long, Andrei Shleifer, Lawrence~H. Summers, and Robert~J.
  Waldmann.
\newblock Noise trader risk in financial markets.
\newblock \emph{Journal of Political Economy}, 98\penalty0 (4):\penalty0
  703--738, 1990.

\bibitem[DeMarzo and Skiadas(1998)]{DeMarzoSkiadas1998}
Peter~M. DeMarzo and Costis Skiadas.
\newblock Aggregation, determinacy, and informational efficiency for a class of
  economies with asymmetric information.
\newblock \emph{Journal of Economic Theory}, 80\penalty0 (1):\penalty0
  123--152, 1998.

\bibitem[Diamond and Verrecchia(1981)]{DiamondVerrecchia1981}
Douglas~W. Diamond and Robert~E. Verrecchia.
\newblock Information aggregation in a noisy rational expectations economy.
\newblock \emph{Journal of Financial Economics}, 9\penalty0 (3):\penalty0
  221--235, 1981.

\bibitem[Glosten and Milgrom(1985)]{GlostenMilgrom1985}
Lawrence~R. Glosten and Paul~R. Milgrom.
\newblock Bid, ask and transaction prices in a specialist market with
  heterogeneously informed traders.
\newblock \emph{Journal of Financial Economics}, 14\penalty0 (1):\penalty0
  71--100, 1985.

\bibitem[Grossman and Stiglitz(1980)]{GrossmanStiglitz1980}
Sanford~J. Grossman and Joseph~E. Stiglitz.
\newblock On the impossibility of informationally efficient markets.
\newblock \emph{American Economic Review}, 70\penalty0 (3):\penalty0 393--408,
  1980.

\bibitem[Heifetz and Polemarchakis(1998)]{HeifetzPolemarchakis1998}
Aviad Heifetz and Heracles~M. Polemarchakis.
\newblock Partial revelation with rational expectations.
\newblock \emph{Journal of Economic Theory}, 80\penalty0 (1):\penalty0
  171--181, 1998.
\newblock \doi{10.1006/jeth.1998.2391}.

\bibitem[Hellwig(1980)]{Hellwig1980}
Martin~F. Hellwig.
\newblock On the aggregation of information in competitive markets.
\newblock \emph{Journal of Economic Theory}, 22\penalty0 (3):\penalty0
  477--498, 1980.

\bibitem[Kasa et~al.(2014)Kasa, Walker, and Whiteman]{KasaWalkerWhiteman2014}
Kenneth Kasa, Todd~B. Walker, and Charles~H. Whiteman.
\newblock Heterogeneous beliefs and tests of present value models.
\newblock \emph{Review of Economic Studies}, 81\penalty0 (3):\penalty0
  1137--1163, 2014.

\bibitem[Kyle(1985)]{Kyle1985}
Albert~S. Kyle.
\newblock Continuous auctions and insider trading.
\newblock \emph{Econometrica}, 53\penalty0 (6):\penalty0 1315--1335, 1985.

\bibitem[Kyle(1989)]{Kyle1989}
Albert~S. Kyle.
\newblock Informed speculation with imperfect competition.
\newblock \emph{Review of Economic Studies}, 56\penalty0 (3):\penalty0
  317--355, 1989.

\bibitem[Mele(2007)]{Mele2007}
Antonio Mele.
\newblock Asymmetric stock market volatility and the cyclical behavior of
  expected returns.
\newblock \emph{Journal of Financial Economics}, 86\penalty0 (2):\penalty0
  446--478, 2007.

\bibitem[Milgrom and Stokey(1982)]{Milgrom1982}
Paul Milgrom and Nancy Stokey.
\newblock Information, trade and common knowledge.
\newblock \emph{Journal of Economic Theory}, 26\penalty0 (1):\penalty0 17--27,
  1982.

\bibitem[Peress(2004)]{Peress2004}
Joel Peress.
\newblock Wealth, information acquisition, and portfolio choice.
\newblock \emph{Review of Financial Studies}, 17\penalty0 (3):\penalty0
  879--914, 2004.

\bibitem[Spiegel and Subrahmanyam(1992)]{SpiegelSubrahmanyam1992}
Matthew Spiegel and Avanidhar Subrahmanyam.
\newblock Informed speculation and hedging in a noncompetitive securities
  market.
\newblock \emph{Review of Financial Studies}, 5\penalty0 (2):\penalty0
  307--329, 1992.

\bibitem[Veldkamp(2011)]{VeldkampBook}
Laura~L. Veldkamp.
\newblock \emph{Information Choice in Macroeconomics and Finance}.
\newblock Princeton University Press, 2011.

\bibitem[Vives(2011)]{Vives2011}
Xavier Vives.
\newblock Strategic supply function competition with private information.
\newblock \emph{Econometrica}, 79\penalty0 (6):\penalty0 1919--1966, 2011.

\bibitem[Wang(1993)]{Wang1993}
Jiang Wang.
\newblock A model of intertemporal asset prices under asymmetric information.
\newblock \emph{Review of Economic Studies}, 60\penalty0 (2):\penalty0
  249--282, 1993.

\bibitem[Wang(1994)]{Wang1994}
Jiang Wang.
\newblock A model of competitive stock trading volume.
\newblock \emph{Journal of Political Economy}, 102\penalty0 (1):\penalty0
  127--168, 1994.

\end{thebibliography}
\end{document}